\documentclass[sigconf,nonacm]{acmart}

\AtBeginDocument{%
  \providecommand\BibTeX{{%
    \normalfont B\kern-0.5em{\scshape i\kern-0.25em b}\kern-0.8em\TeX}}}

\copyrightyear{2021}
\acmYear{2021}
\setcopyright{acmcopyright}\acmConference[CIKM '21]{Proceedings of the 30th ACM International Conference on Information and Knowledge Management}{November 1--5, 2021}{Virtual Event, QLD, Australia}
\acmBooktitle{Proceedings of the 30th ACM International Conference on Information and Knowledge Management (CIKM '21), November 1--5, 2021, Virtual Event, QLD, Australia}
\acmPrice{15.00}
\acmDOI{10.1145/3459637.3482408}
\acmISBN{978-1-4503-8446-9/21/11}

\usepackage{amsmath}
\usepackage{amsfonts}
\usepackage{amsthm}
\usepackage{latexsym}   
\usepackage{graphicx}
\usepackage{color}
\usepackage{hyperref}
\usepackage[linesnumbered,ruled,vlined]{algorithm2e}
\usepackage{booktabs}
\usepackage{bm}
\usepackage{subfig}

\usepackage{CJKutf8}

\usepackage{cleveref}
\usepackage{autonum}

\DeclareMathOperator{\Diag}{Diag}
\DeclareMathOperator{\proj}{proj}
\DeclareMathOperator{\prox}{prox}

\newcommand{\R}{\mathbb R}
\newcommand{\1}{\mathbf 1}
\newcommand{\0}{\mathbf 0}

\usepackage{mathtools}
\DeclarePairedDelimiter\ceil{\lceil}{\rceil}

\DeclarePairedDelimiterX\inner[2]{\langle}{\rangle}{{#1},{#2}}

\DeclarePairedDelimiter\norm{\|}{\|}
\DeclarePairedDelimiter\set{\{}{\}}
\DeclarePairedDelimiter\prn{(}{)}

\DeclarePairedDelimiterX\Set[2]{\{}{\}}{\mspace{2mu}{#1}:{#2}\mspace{2mu}}

\newcommand{\argmax}{\mathop{\rm argmax}}
\newcommand{\argmin}{\mathop{\rm argmin}}

\newtheorem{problem}{Problem}

\begin{document}

\fancyhead{}
\title{A Projected Gradient Method for Opinion Optimization with Limited Changes of Susceptibility to Persuasion}



\author{Naoki Marumo}
\affiliation{%
  \institution{University of Tokyo \& NTT CS Labs.}
  \city{Tokyo}
  \country{Japan}
}
\email{marumo-naoki@g.ecc.u-tokyo.ac.jp}

\author{Atsushi Miyauchi}
\affiliation{%
  \institution{University of Tokyo}
  \city{Tokyo}
  \country{Japan}
}
\email{miyauchi@mist.i.u-tokyo.ac.jp}

\author{Akiko Takeda}
\affiliation{%
  \institution{University of Tokyo \& RIKEN AIP}
  \city{Tokyo}
  \country{Japan}
}
\email{takeda@mist.i.u-tokyo.ac.jp}

\author{Akira Tanaka}
\affiliation{%
  \institution{University of Tokyo}
  \city{Tokyo}
  \country{Japan}
}
\email{akira@cl.rcast.u-tokyo.ac.jp}

%

\begin{abstract}
Many social phenomena are triggered by public opinion that is formed in the process of opinion exchange among individuals. To date, from the engineering point of view, a large body of work has been devoted to studying how to manipulate individual opinions so as to guide public opinion towards the desired state. Recently, Abebe et al. (KDD 2018) have initiated the study of the impact of interventions at the level of susceptibility rather than the interventions that directly modify individual opinions themselves. For the model, Chan et al. (The Web Conference 2019) designed a local search algorithm to find an optimal solution in polynomial time. However, it can be seen that the solution obtained by solving the above model might not be implemented in real-world scenarios. In fact, as we do not consider the amount of changes of the susceptibility, it would be too costly to change the susceptibility values for agents based on the solution. 

In this paper, we study an opinion optimization model that is able to limit the amount of changes of the susceptibility in various forms. First we introduce a novel opinion optimization model, where the initial susceptibility values are given as additional input and the feasible region is defined using the $\ell_p$-ball centered at the initial susceptibility vector. For the proposed model, we design a projected gradient method that is applicable to the case where there are millions of agents. Finally we conduct thorough experiments using a variety of real-world social networks and demonstrate that the proposed algorithm outperforms baseline methods. 
\end{abstract}




\maketitle

\section{Introduction}\label{sec:intro}

Many social phenomena are triggered by public opinion that is formed in the process of opinion exchange among individuals.  
Opinion dynamics, a mathematical model that captures such a fusion process, has played a key role in social psychology and related fields~\cite{Acemoglu+11,Dong+18,Hegselmann+02}. 
To date, from the engineering point of view, a large body of work has been devoted to studying how to manipulate individual opinions so as to guide public opinion towards the desired state~\cite{Kempe+03,Domingos+01,Gionis+13}; 
most existing studies consider interventions that directly act on individual opinions (see e.g., \cite{Kempe+03,Domingos+01,Gionis+13} and references therein). 

Recently, Abebe et al.~\cite{Abebe+18} have initiated the study of the impact of interventions at the level of susceptibility rather than the interventions that directly modify the opinions themselves. 
Their model is inspired by DeGroot~\cite{DeGroot74} and Friedkin and Johnsen~\cite{Friedkin+90}. 
Let $G=(V,E)$ be a simple undirected graph that represents interactions among agents. 
Each agent $i\in V$ has an innate opinion $s_i\in [0,1]$, where a higher $s_i$ value means that the agent $i$ has a more favorable opinion towards a given topic. 
The susceptibility to persuasion of agents are captured by resistance values $\alpha_i\in (0,1]$, 
where a higher $\alpha_i$ value implies that the agent $i$ is less susceptible to changing his or her opinion. 
The opinion dynamics based on the innate opinions and the resistance values evolve in a discrete-time fashion 
according to the interactions defined by the graph $G$. 
The final opinions of the agents are determined by computing the equilibrium state of the above dynamics. 
The goal of their model is to optimize (minimize or maximize) the sum of final opinions of the agents 
by appropriately determining a resistance value $\alpha_i$ in the interval $[l_i,u_i]$ for each agent $i\in V$. 

For the above model (especially the minimization version), 
Abebe et al.~\cite{Abebe+18} have attempted to give a polynomial-time exact algorithm by exploiting the convexity of the objective function. 
Later Chan et al.~\cite{Chan+19} pointed out that the objective function is neither convex nor concave 
and hence one cannot directly use a general convex programming approach to design an effective algorithm. 
To overcome this issue, Chan et al.~\cite{Chan+19} thoroughly analyzed the structural properties of the objective function. 
Then they showed that any optimal solution satisfies $\alpha_i = l_i$ or $\alpha_i=u_i$ for every $i\in V$, 
implying the combinatorial nature of the model. 
Moreover, they found the fact that any local optimum is also a global optimum, in spite of non-convexity, 
when defining the neighbor of a solution as the set of solutions generated by switching $\alpha_i$ from $l_i$ to $u_i$ or from $u_i$ to $l_i$. 
Based on this fact, they finally designed a local search algorithm (for a slight generalization of the model by Abebe et al.~\cite{Abebe+18}) to find an optimal solution in polynomial time. 

In the opinion optimization models studied in Abebe et al.~\cite{Chan+19} and Chan et al.~\cite{Chan+19}, 
we are allowed to determine the resistance values freely as long as they satisfy lower and upper bound constraints. 
The lower and upper bounds are useful to represent the personality of agents: 
For agents with less (resp. much) willingness to change their susceptibility, we can set tight (resp. loose) lower and upper bounds on the initial resistance values estimated. 
However, it should be emphasized that the solution obtained by solving the above model might not be implemented in real-world scenarios. 
In fact, as we do not consider the amount of changes of the susceptibility, it would be too costly to change the resistance values for agents based on the solution. 

In this paper, we study an opinion optimization model that is able to limit the amount of changes of the susceptibility in various forms.  
First we introduce a novel opinion optimization model, where the initial susceptibility values are given as additional input 
and the feasible region is defined using the $\ell_p$-ball ($p\geq 0$) centered at the initial susceptibility vector. 
Note that Abebe et al.~\cite{Abebe+18} also considered a budgeted variant of their model, which corresponds to our model with $p=0$. 
That is, one has the $\ell_0$-norm constraint, just allowing to modify the resistance values of at most $k$ agents. 
Abebe et al.~\cite{Abebe+18} proved that this problem is NP-hard, in contrast to the original problem, and gave a simple heuristic. 
Moreover, very recently, 
Chan and Lee~\cite{Chan+21} introduced a variant with the $\ell_1$-norm constraint, which is equivalent to ours with $p=1$, inspired by various real-world scenarios. 
They showed that this problem is NP-hard by constructing a sophisticated polynomial-time reduction from the vertex cover problem. 
However, there is still a lack of effective algorithms even for this special case. 

For our general model with $p\geq 1$, we design a projected gradient method that is applicable to the case where there are millions of agents. 
Note that as our model has the $\ell_p$-norm constraint over all agents, there is no longer the above-mentioned combinatorial nature, 
which suggests using such a continuous optimization approach. 
An ordinary projected gradient method requires the Lipschitz continuity of the gradient of the objective function over the entire space, 
i.e., $\mathbb{R}^V$ for our model. 
However, as will be shown later, our problem satisfies the Lipschitz continuity only over the feasible region. 
This gap can be filled in by employing Nesterov's variant of the projected gradient methods~\cite{nesterov2013gradient}. 

To design an efficient projected gradient method, we need to be able to compute the gradient of the objective function efficiently. 
However, as presented by Abebe et al.~\cite{Abebe+18} and Chan et al.~\cite{Chan+19} (and will be shown later), 
the direct computation of the gradient involves the matrix inverse, 
which leads to the cubic time complexity. 
To overcome this issue, we construct a reduction from the computation of the gradient to solving some linear systems, 
enabling us to reduce the time complexity drastically. 
It should be noted that the computation of the objective function value also involves the matrix inverse. 
Our technique is again applicable to it. 

Finally we conduct thorough computational experiments using a variety of real-world social networks 
and demonstrate that the proposed algorithm outperforms baseline methods.

\section{Related Work}\label{sec:related}

The opinion dynamics models studied in Abebe et al.~\cite{Abebe+18} and Chan et al.~\cite{Chan+19} 
are inspired by DeGroot~\cite{DeGroot74} and Friedkin and Johnsen~\cite{Friedkin+99}. 
The seminal work by DeGroot~\cite{DeGroot74} introduced a mathematical model of consensus reaching 
wherein there are agents with opinions on some topic and each agent's opinion is updated using the average value of his or her neighbors' opinions in a given social network. 
Friedkin and Johnsen~\cite{Friedkin+99} considered a generalization of DeGroot's model, 
where each agent has an additional attribute called the innate belief and each agent's opinion is determined by mixing it into the average value of its neighbors' opinions. 
These models are continuous model, where each agent has a continuous opinion. 
On the other hand, there also exist discrete models, as in Domingos and Richardson~\cite{Domingos+01} and Kempe et al.~\cite{Kempe+03}. 

Projected gradient method is a well-known iterative method for constrained continuous optimization problems, 
which is based on the strategy that computes the gradient of the objective function at the current solution, 
gets the next solution (not necessarily feasible) using the gradient, projects the solution onto the feasible region, 
and then repeats these procedures until some termination criterion is met.
The main computation cost is in computing the gradient and the projection operator.
Various examples of the projection operator that can be easily derived, for example, by the closed-form, are shown in Beck~\cite{beck2017first}.
As long as the objective function is possibly non-convex but smooth enough and the feasible region is convex, the convergence to an optimal solution
(resp. a stationary point) is ensured for the convex (resp. non-convex) objective function, and the
convergence speed is also evaluated.

Non-convex optimization is a key primitive in mathematical optimization theory, 
which has gained increasing attention recently in various fields.
Application examples include gene expression analysis, recommendation systems, clustering, and anomaly detection~\cite{jain2017nonconv}. 
In particular, in the context of machine learning, 
non-convex optimization methods are actively used for hyperparameter optimization of neural networks~\cite{HaoNIPS18,jain2017nonconv}.

\section{Model}\label{sec:model}

In this section, we first revisit the opinion optimization model considered by Chan et al.~\cite{Chan+19}, 
which is a slight generalization of the model proposed by Abebe et al.~\cite{Abebe+18}, 
and then formally describe the model that we will study. 

Let $V$ be the (finite) set of agents whose innate opinions are given by a vector $\bm{s}\in [0,1]^V$, 
where a higher $s_i$ value means that the agent $i$ has a more favorable opinion towards a given topic. 
The interactions among agents are captured by a row stochastic matrix $P\in [0,1]^{V\times V}$, 
that is, each element of $P$ is nonnegative and each row-sum equals one. 
As in Chan et al.~\cite{Chan+19}, we call this matrix $P$ the interaction matrix. 

Note that in the model proposed by Abebe et al.~\cite{Abebe+18}, we are given a simple undirected graph $G=(V,E)$ 
that represents interactions among agents, rather than the interaction matrix itself. 
Then they constructed the interaction matrix as follows: 
For any $(i,j)\in V\times V$, set $P_{ij}=1/\text{deg}(i)$ if there exists $\{i,j\}\in E$ and $P_{ij}=0$ otherwise, where $\text{deg}(i)$ denotes the degree of $i\in V$. 
Therefore, we see that the model by Chan et al.~\cite{Chan+19} is a generalization of the model by Abebe et al.~\cite{Abebe+18}. 

The susceptibility to persuasion of agents are represented by the resistance vector $\bm{\alpha}\in (0,1]^V$, 
where a higher $\alpha_i$ value implies that the agent $i$ is less susceptible to changing his or her opinion. 

The opinion dynamics evolve in a discrete-time fashion. 
Specifically, the opinion vector $\bm{z}^{(t+1)}$ at time $t+1$ ($t=0,1,\dots$) is determined using $\bm{z}^{(t)}$ as follows: 
\begin{align}
  \bm{z}^{(0)}= \bm{s},\quad \bm z^{(t+1)}= \Diag(\bm \alpha) \bm s + \Diag(\1 - \bm \alpha) P \bm z^{(t)},
\end{align}
where $\Diag(\bm \alpha)$ is the diagonal matrix with $(\Diag(\bm \alpha))_{ii}=\alpha_i$. 
It is easy to see that the equilibrium opinion vector is given by 
\begin{align}
  \bm z^{(\infty)}
  = (I - \Diag(\1 - \bm \alpha) P)^{-1} \Diag(\bm \alpha) \bm s,
  \label{eq:def_of_z}
\end{align}
which exists if $\alpha_i > 0$ for every $i\in V$. 
Therefore, the sum of equilibrium opinions (as a function of the resistance vector) is given by 
\begin{align}
  f(\bm{\alpha})
  \coloneqq \bm{1}^\top \bm z^{(\infty)}
  = \bm{1}^\top (I - \Diag(\1 - \bm \alpha) P)^{-1} \Diag(\bm \alpha) \bm s. 
\end{align}

In the opinion optimization model considered by Chan et al.~\cite{Chan+19}, we are asked to find a resistance vector $\bm{\alpha}\in \mathbb{R}^V$ with $\bm{l}\leq \bm{\alpha}\leq \bm{u}$ that minimizes $f(\bm{\alpha})$.
Here, the inequality constraints are elementwise, and $\bm{l}$ and $\bm{u}$ are lower and upper bound vectors, respectively, that satisfy $\bm{0}<\bm{l}\leq \bm{u}\leq \bm{1}$.
For convenience, throughout this paper, we write
\begin{align}
[\bm{l},\bm{u}]
\coloneqq \{\bm{\alpha}\in \mathbb{R}^V: \bm{l}\leq \bm{\alpha}\leq \bm{u}\}. 
\end{align}
Note that one can also consider the maximization variant of the problem. 
However, Chan et al.~\cite{Chan+19} decided that they focus only on the minimization version because techniques needed to analyze the maximization variant are quite similar. 

Chan et al.~\cite{Chan+19} analyzed structural properties of the objective function. 
They first showed that the function is neither convex nor concave, contrary to the claim by Abebe et al.~\cite{Abebe+18} that the function is convex. 
On the other hand, they found the fact that any local optimum is also a global optimum, in spite of its non-convexity. 
Using this fact, they designed a local search algorithm to find an optimal solution in polynomial time. 

In the opinion optimization models considered by Abebe et al.~\cite{Abebe+18} and Chan et al.~\cite{Chan+19}, 
we are allowed to determine the resistance vector $\bm{\alpha}\in (0,1]^{V}$ freely 
as long as it satisfies lower and upper bound constraints, i.e., $\bm{\alpha}\in [\bm{l}, \bm{u}]$. 
The lower and upper bound vectors are useful to represent the personality of agents: 
For agents with less (resp. much) willingness to change their susceptibility, we can set tight (resp. loose) lower and upper bounds 
on the initial resistance values estimated. 
However, it should be emphasized that the solution obtained by solving the above model might not be implemented in real-world scenarios. 
In fact, as we do not consider the amount of changes of the susceptibility, it would be too costly to change the resistance values for agents based on the solution. 

To overcome this issue, we introduce an opinion optimization model that is able to limit the amount of changes of the susceptibility in various forms. 
Specifically, our model can be defined as follows: 
\begin{problem}[Opinion optimization with $\ell_p$-norm constraint]
Let $p\geq 0$. 
Given the set $V$ of agents, the innate opinion vector $\bm{s}\in [0,1]^{V}$, the interaction matrix $P\in [0,1]^{V\times V}$, 
the initial resistance vector $\bm{\alpha}^\mathrm{init}$ and its lower and upper bound vectors $\bm{l}$ and $\bm{u}$, respectively,  
such that $\bm{0}<\bm{l}\leq \bm{\bm{\alpha}}^\mathrm{init}\leq \bm{u}\leq \bm{1}$, and $k\in \mathbb{R}_{\geq 0}$, 
we are asked to find a resistance vector $\bm{\alpha}\in [\bm{l}, \bm{u}]$ that minimizes $f(\bm{\alpha})$ under the constraint that $\|\bm{\alpha}-\bm{\alpha}^\mathrm{init}\|_p\leq k$. 

For reference, we give its mathematical programming formulation: 
\begin{alignat}{3}
&\mathrm{(P)}&\ \    &\mathrm{minimize}_{\bm{\alpha}\in \mathbb{R}^V} &\ \ &f(\bm{\alpha}) = \bm{1}^\top (I - \Diag(\1 - \bm \alpha) P)^{-1} \Diag(\bm \alpha) \bm s\\
&&   &\mathrm{subject\ to} &  &\|\bm{\alpha}-\bm{\alpha}^\mathrm{init}\|_p\leq k,\quad \bm{\alpha}\in [\bm{l},\bm{u}].
\end{alignat}
We denote by $\mathcal{C}$ the feasible region of the problem, i.e., 
\begin{align}
\mathcal{C} \coloneqq \{\bm{\alpha}\in \mathbb{R}^V : \|\bm{\alpha}-\bm{\alpha}^\mathrm{init}\|_p\leq k, \ \bm{\alpha}\in [\bm{l},\bm{u}]\}. 
\end{align}
\end{problem}

As in Chan et al.~\cite{Chan+19}, we can also consider the maximization variant of our model. 
However, consistent with them, we focus only on the minimization version. 
Our technical contributions can easily be duplicated for the maximization variant. 

Note that Abebe et al.~\cite{Abebe+18} also considered a budgeted variant of their model, which corresponds to our model with $p=0$.
That is, we have the constraint $\|\bm{\alpha}-\bm{\alpha}^\mathrm{init}\|_0= |\{i\in V: \alpha_i\neq \alpha^\mathrm{init}_i\} |\leq k$, 
just allowing to modify the resistance values of at most $k$ agents. 
Abebe et al.~\cite{Abebe+18} proved that this problem is NP-hard and gave a simple heuristic.
Moreover, very recently, 
Chan and Lee~\cite{Chan+21} introduced a variant with the $\ell_1$-norm constraint, i.e., $\|\bm{\alpha}-\bm{\alpha}^\mathrm{init}\|_1=\sum_{i\in V}|\alpha_i-\alpha^\text{init}_i|\leq k$, 
which is equivalent to ours with $p=1$. 
They showed that this problem is NP-hard by constructing a polynomial-time reduction from the vertex cover problem.
However, there is still a lack of effective algorithms even for this special case.

\section{Algorithm}\label{sec:algorithm}
In this section, we propose a projected gradient method for Problem~(P) with $p\geq 1$. 

\subsection{Projected Gradient Method}

Our method, which we describe formally as Algorithm~\ref{alg:PGM},
iteratively executes two operations: calculation of the gradient vector and projection calculation onto the feasible region
until a termination criterion is satisfied. 
The two operations are executed in \Cref{line:update_alpha}, where after $\bm \alpha^t - \eta \nabla f(\bm \alpha^t)$ is computed with some stepsize $\eta$,
it is projected onto $\mathcal{C}$ using 
the projection operator $\proj_{\mathcal C}$.
Here the projection operator $\proj_{\mathcal S}: \R^V \to \R^V$
for a closed convex set $\mathcal S \subseteq \R^V$ is defined by
\begin{align}
  \proj_{\mathcal S}(\bm \alpha)
  \coloneqq \argmin_{\bm \beta \in \mathcal S} \|\bm \beta - \bm \alpha\|_2.
\end{align}
It should be noted that $\bm \alpha^t \in \mathcal C$  because of the update rule in \Cref{line:update_alpha} and
the objective value $f(\bm \alpha^t)$ strictly decreases in every iteration because of \Cref{line:sufficient_decrease_cond}.
The key computations in \Cref{alg:PGM} are $\nabla f(\bm \alpha)$ and $\proj_{\mathcal C}(\bm \alpha)$ for a given $\bm \alpha$.
We will discuss how to efficiently compute them in \Cref{sec:eval_gradient,sec:projection}, respectively.

This algorithm employs a backtracking rule to select a suitable stepsize $\eta$ by \Cref{line:dec_eta}.
A constant stepsize $\eta = 1/L$ can be used for the method 
with the Lipschitz constant\footnote{
  Since there exists the Hessian matrix of $f$ on $\mathcal C$
  and the set $\mathcal C$ is closed and bounded,
  $\nabla f$ is Lipschitz continuous on $\mathcal C$.
}
$L$ of $\nabla f$,
but it is well known that the backtracking rule  performs better by finding a larger stepsize than the constant one in most cases.


In \Cref{alg:PGM}, the stepsize $\eta$ may decrease but never increases.
For practical speedup of the algorithm,
we can introduce a heuristic technique that increases $\eta$ at the end of each iteration of \Cref{line:outer_loop},
i.e., insert the step ``$\eta \gets \gamma_{\mathrm{inc}} \eta$'' for $\gamma_{\mathrm{inc}} > 1$ just after \Cref{line:return_ans}.
This also mitigates the effect of the initial stepsize $\eta_0$ on the convergence speed
and makes it easier to tune the hyperparameter $\eta_0$.
In the implementation used in our experiments, we employed such a technique.

The termination criterion of Algorithm~\ref{alg:PGM} is described in \Cref{line:termination_cond}.
The condition is equivalent to $\norm{\mathcal G_\eta(\bm \alpha^t)} \leq \epsilon$, where
 the gradient mapping $\mathcal G_\eta: \R^V \to \R^V$ of $f$ on $\mathcal C$ is defined with $\eta > 0$  by 
\begin{align}
  \mathcal G_\eta(\bm \alpha) \coloneqq \frac{1}{\eta} (\bm \alpha - \bm \alpha^+),
  \ \ \text{where}\ \ 
  \bm \alpha^+ \coloneqq \proj_{\mathcal C}(\bm \alpha - \eta \nabla f(\bm \alpha)).
\end{align}
Therefore, the output $\bm \alpha \in \mathcal C$ of \Cref{alg:PGM} is guaranteed to satisfy $\norm{\mathcal G_\eta(\bm \alpha)} \leq \epsilon$
if the algorithm terminates. The gradient mapping can be regarded as the generalized gradient for constrained problems; indeed,
it is known that $\mathcal G_\eta(\widehat{\bm \alpha}) = \0$ if and only if $\widehat{\bm \alpha}$ is a stationary point of Problem (P).
This statement holds for any $\eta > 0$.
Hence, the gradient mapping 
is a popular measure used for termination criteria of the first-order methods, including the projected gradient method.

\begin{algorithm}[t]
  \caption{Projected gradient method with backtracking}
  \label{alg:PGM}
    \SetKwInOut{Input}{Input}
    \SetKwInOut{Output}{Output}
    \Input{\ 
      $\bm \alpha^0 \in \mathcal C$, \ 
      $\eta_{0} > 0$, \ 
      $\gamma < 1$, \ 
      $\epsilon > 0$
    }
    \Output{\ 
      $\bm \alpha \in \mathcal C$ such that $\norm{\mathcal G_\eta(\bm \alpha)} \leq \epsilon$ for some $\eta \in (0, \eta_{0}]$\hspace{-3em}
    }
    $\eta \gets \eta_{0}$\;
    \For{
      $t = 1,2,\dots$
      \label{line:outer_loop}
    }{
      \While{
        \textnormal{TRUE}
        \label{line:inner_loop}
      }{
        $\bm \alpha^{t} \gets \proj_{\mathcal C}(\bm \alpha^{t-1} - \eta \nabla f(\bm \alpha^{t-1}))$ \label{line:update_alpha}\;
        \tcp{$\proj_{\mathcal C}$ is computed by \Cref{alg:projection}.}
        \If{$f(\bm \alpha^{t}) \leq f(\bm \alpha^{t-1}) - \frac{1}{2\eta}\|\bm \alpha^{t} - \bm \alpha^{t-1}\|_2^2$ \label{line:sufficient_decrease_cond}}{
          \textbf{break}\label{line:success_break}\; 
        }
        \Else{
        $\eta \gets \gamma \eta$ \label{line:dec_eta}\;
        }
      }
      \If{$\norm{\bm \alpha^{t} - \bm \alpha^{t-1}}_2 \leq \eta \epsilon$ \label{line:termination_cond}}{
        \Return $\bm \alpha^{t-1}$\label{line:return_ans}\;
      }
    }
\end{algorithm}

\subsection{Finite Termination and Iteration Complexity}
Here we discuss theoretical properties of Algorithm~\ref{alg:PGM} such as the finite termination and the worst-case iteration complexity.
Similar theoretical guarantees
can be found in literature~\cite{beck2017first,nesterov2013gradient},
but to the best of our knowledge, there are no proofs with the same assumptions and backtracking rule as ours.
Therefore, for self-containedness, we provide complete proofs of our results.

Let $L$ be the Lipschitz constant of $\nabla f$ on $\mathcal C$. We have the following lemma: 
\begin{lemma}
  \label{lem:unsuccessful}
  In \Cref{alg:PGM},
  \begin{itemize}
    \item
    the parameter $\eta$ always satisfies $\eta \geq \min \set{\eta_0, \gamma/L}$, and
    \item 
    \Cref{line:dec_eta} is executed at most $\max\set{ \ceil{\log_{1/\gamma}(\eta_0 L)} , 0}$ times.
  \end{itemize}
\end{lemma}
\begin{proof}
  By Equations (2.17) and (2.20) in \cite{nesterov2013gradient},
  the condition in \Cref{line:sufficient_decrease_cond} of \Cref{alg:PGM} must hold if $\eta \leq 1/L$.
  Therefore, if $\eta_0 \leq 1/L$,
  \Cref{line:dec_eta} will not be executed and
  the parameter $\eta$ always satisfies $\eta = \eta_0$.
  Otherwise \Cref{line:dec_eta} may be executed but the number of executions does not exceed
  $\ceil{\log_{1/\gamma}(\eta_0 L)}$, and $\eta$ always satisfies $\eta > \gamma / L$.
  This completes the proof. 
\end{proof}

The second result of the above lemma implies that
the loop of \Cref{line:inner_loop} in \Cref{alg:PGM} terminates after a finite number of iterations.

\begin{theorem}\label{theo:iter}
  Let $\bm \alpha^* \in \mathcal C$ be an optimal solution to Problem (P).
  \begin{itemize}
    \item 
    \Cref{alg:PGM} terminates after at most
    \[
      \ceil[\bigg]{ \frac{2 ( f(\bm \alpha^0) - f(\bm \alpha^*) )}{ \epsilon^2 \min \set{\eta_0, \gamma/L}}}
      = O(\epsilon^{-2})
    \]
    iterations of \Cref{line:outer_loop}
    and outputs a point $\bm \alpha \in \mathcal C$ such that
    $\norm{\mathcal G_\eta(\bm \alpha)} \leq \epsilon$ for some $\eta \in (0, \eta_0]$.
    \item
    Let $\{\bm \alpha^t\}_{t= 0}^\infty$ be an infinite sequence generated by Algorithm~\ref{alg:PGM} with $\epsilon=0$.
    Then, any accumulation point of $\{\bm \alpha^t\}_{t= 0}^\infty$ is a stationary point of Problem (P).
  \end{itemize}
\end{theorem}

\begin{proof}
  For $t \geq 0$,
  let $\eta_t$ be the value of the parameter $\eta$ of \Cref{alg:PGM}
  at the end of the $t$-th iteration of \Cref{line:outer_loop}.
  By the condition in \Cref{line:sufficient_decrease_cond}, the definition of $\mathcal G_\eta(\bm \alpha)$, and \Cref{lem:unsuccessful}, we have
  \begin{align}
    f(\bm \alpha^{t}) - f(\bm \alpha^{t-1})
    &\leq - \frac{\eta_t}{2} \| \mathcal G_{\eta_t}(\bm \alpha^{t-1}) \|^2 \\
    &\leq - \frac{\min \set{\eta_0, \gamma/L}}{2} \| \mathcal G_{\eta_t}(\bm \alpha^{t-1}) \|^2.
    \label{eq:diff_obj} 
  \end{align}
  Summing up these inequalities for $t = 1,2,\dots,T$, we obtain
  \begin{align}
    \min_{1 \leq t \leq T} \!
    \| \mathcal G_{\eta_t}(\bm \alpha^{t-1}) \|^2
    &\leq \frac{2 (f(\bm \alpha^0) - f(\bm \alpha^{T})) }{T \min \set{\eta_0, \gamma/L} }
    \leq \frac{2 (f(\bm \alpha^0) - f(\bm \alpha^*)) }{T \min \set{\eta_0, \gamma/L} }.
  \end{align}
  The first result follows from the equivalence between the termination condition in \Cref{line:termination_cond} and $\| \mathcal G_{\eta_t}(\bm \alpha^{t-1}) \| \leq \epsilon$.
  Since $\|\mathcal G_\eta(\bm \alpha)\|$ is decreasing in $\eta$ (see Lemma 2 in \cite{nesterov2013gradient}), 
  we have from Equation~\eqref{eq:diff_obj} and $\eta_t \leq \eta_0$ that
  \begin{align}
    \sum_{t = 1}^\infty \| \mathcal G_{\eta_0}(\bm \alpha^{t-1}) \|^2
    \leq
    \sum_{t = 1}^\infty \| \mathcal G_{\eta_t}(\bm \alpha^{t-1}) \|^2
    \leq \frac{2 (f(\bm \alpha^0) - f(\bm \alpha^*)) }{\min \set{\eta_0, \gamma/L} }.
  \end{align}
  Thus, we obtain $\lim_{t \to \infty} \| \mathcal G_{\eta_0}(\bm \alpha^t) \| = 0$, which implies the second result.
\end{proof}



\subsection{Evaluation of Gradient of Objective Function}
\label{sec:eval_gradient}
The following lemma gives the expression of the gradient vector of the objective function in Problem (P).
\begin{lemma}
  Let $M = I - \Diag(\1 - \bm \alpha) P$.
  Then it holds that 
  \begin{align}
    \nabla f(\bm{\alpha})
    = \Diag \prn[\big]{M^{-\top} \1} \prn[\big]{\bm{s} - P M^{-1} \Diag(\bm{\alpha}) \bm{s} }, 
    \label{eq:grad_f}
  \end{align}
where $M^{-\top}$ denotes the transposed inverse of $M$, i.e., $M^{-\top}= (M^{-1})^\top$. 
\end{lemma}
\begin{proof}
  Let $P = [\bm{p}_1,\dots,\bm{p}_n]^\top$ and $\bm v = M^{-\top} \bm 1$. 
  As $f(\bm{\alpha}) = \bm{v}^\top \Diag(\bm{\alpha}) \bm{s}$ holds, we have 
  \begin{align}
    \frac{\partial f}{\partial \alpha_i}
    = v_i s_i
    + \prn[\Big]{\frac{\partial \bm{v}}{\partial \alpha_i}}^\top \Diag(\bm \alpha) \bm s.
    \label{eq:f_partial_alpha_i}
  \end{align}
  On the other hand, differentiating both sides of $\bm{v}^\top M  = \bm 1^\top$ with respect to $\alpha_i$, we have 
$
    \bm{v}^\top \frac{\partial M}{\partial \alpha_i}
    + \prn[\Big]{ \frac{\partial \bm{v}}{\partial \alpha_i} }^\top M
    = \0^\top, 
$
  which implies that 
  \begin{align}
    \prn[\Big]{ \frac{\partial \bm{v}}{\partial \alpha_i} }^\top
    = - \bm{v}^\top \frac{\partial M}{\partial \alpha_i} M^{-1}
    = - v_i \bm{p}_i^\top M^{-1}. 
    \label{eq:vt_partial_alpha_i}
  \end{align}
  Combining Equations \eqref{eq:f_partial_alpha_i} and \eqref{eq:vt_partial_alpha_i}, we have
  \[
    \frac{\partial f}{\partial \alpha_i}
    = v_i \prn[\big]{
      s_i - \bm{p}_i^\top M^{-1} \Diag(\bm \alpha) \bm s
    }, 
  \]
  which is equivalent to Equation~\eqref{eq:grad_f}. 
\end{proof}

The expression $\nabla f(\bm{\alpha})$ includes the inverse matrices $M^{-\top}$ and $M^{-1}$. 
However, we can avoid the inverse matrix computation consuming $O(|V|^3)$ time, by solving 
each linear equation having the matrix $M$ or $M^\top$ as a coefficient matrix once. 
The sparsity of  $M$ and $M^\top$ can be utilized by solving those equations using the Biconjugate gradient (BiCG) method.
Since the product of matrix $ M, M^\top$ and a vector can be calculated in $ O (| V | + | E |) $ time, 
the gradient $\nabla f(\bm{\alpha})$ can be calculated in $O(T (|V| + |E|))$ time, 
where $T$ is the number of iterations in the BiCG method. 


\subsection{Projection onto Feasible Region}
\label{sec:projection}
The feasible region $\mathcal C$ consists of two kinds of constraints: the box constraint and the $\ell_p$-norm constraint,
which makes 
the projection operation seem difficult at first glance. 
However, we can show that the projection operation can be executed by a simple bisection method. 

Here we introduce some notation to describe an algorithm and its analysis. 
We define three functions $g_1, g_2, g: \R^V \to \R \cup \set{+\infty}$ by
\begin{align}
  g_1(\bm \alpha) &\coloneqq \|\bm \alpha - \bm \alpha^{\mathrm{init}}\|_p^p,
  \label{eq:def_of_g1}\\
  g_2(\bm \alpha) &\coloneqq
  \begin{dcases*}
    0 & if $\bm \alpha \in [\bm l, \bm u]$,\\
    +\infty & otherwise,
  \end{dcases*}\quad 
  g(\bm \alpha) \coloneqq g_1(\bm \alpha) + g_2(\bm \alpha).
\end{align}
For a convex function $h: \R^V \to \R \cup \set{+\infty}$, its proximal operator
$\prox_h: \R^V \to \R^V$ is defined by
\begin{align}
  \prox_h(\bm \alpha)
  \coloneqq \argmin_{\bm \beta \in \R^V} \set[\Big]{
    h(\bm \beta) + \frac{1}{2} \|\bm \beta - \bm \alpha\|_2^2
  }.
\end{align}
Note that the projection operator $\proj$ onto a set $\mathcal S$ can be rewritten as the proximal operator with the indicator function of $\mathcal S$ such as $\proj_{[\bm l, \bm u]}=\prox_{g_2}$. Then we have the following theorem:

\begin{theorem}
  [Theorem 6.30 in \cite{beck2017first}, modified to our setting]
  \label{thm:bisection}
  There exists $\lambda^* \geq 0$ such that
  $\prox_{\lambda g}(\bm \alpha) \notin \mathcal C$ for $\lambda \in [0, \lambda^*)$ and
  $\prox_{\lambda g}(\bm \alpha) \in \mathcal C$ for $\lambda \in [\lambda^*, +\infty)$.
  Furthermore, $\proj_{\mathcal C}(\bm \alpha) = \prox_{\lambda^* g}(\bm \alpha)$ holds for such $\lambda^*$.
\end{theorem}
\begin{proof}
  We rewrite the feasible region $\mathcal{C}$ as the sublevel set of the convex function $g$, 
  i.e., $\mathcal C = \Set{\bm \alpha \in \R^V }{ g(\bm \alpha) \leq k^p }$,
  and apply Theorem 6.30 in \cite{beck2017first}.
\end{proof}

From the above theorem, it suffices to apply a
bisection algorithm to find the smallest $\lambda^*$ for $\prox_{\lambda g}(\bm \alpha) \in \mathcal C$.
We can use the following theorem for computing $\prox_{\lambda g}(\bm \alpha)$. 
\begin{theorem}
  \label{thm:decompose_prox}
  For all $\bm \alpha \in \R^V$ and $\lambda > 0$,
  it holds that
  $\prox_{\lambda g}(\bm \alpha) = \proj_{[\bm l, \bm u]} (\prox_{\lambda g_1} (\bm \alpha))$.
\end{theorem}
\begin{proof}
  Since both $g_1$ and $g_2$ are decomposable into the sum of one-variable convex functions,
  the proximal operators $\prox_{\lambda g_1}$, $\prox_{\lambda g_2}$, and $\prox_{\lambda g}$
  are also decomposable for each variable.
  Thus, by Theorem 2 in \cite{yu2013decomposing}, we have 
  $\prox_{\lambda g}(\bm \alpha) = \prox_{\lambda g_2} (\prox_{\lambda g_1} (\bm \alpha))$
  for all $\lambda \geq 0$.
  We also have $\proj_{[\bm l, \bm u]} = \prox_{\lambda g_2}$ for all $\lambda > 0$,
  and obtain the desired results.
\end{proof}

From the above discussion, we obtain an algorithm for calculating the projection onto the feasible region, 
presented in \Cref{alg:projection}.
In \Cref{line:check_feasiblity,line:return_projected_vec} of \Cref{alg:projection}, 
the calculation of $\prox_{\lambda g_1}$ for $\lambda > 0$ is required.
Since $g_1$ can be expressed as the sum of one-variable convex functions,
the proximal operator of $\lambda g_1$ can be easily calculated.
In particular, the closed formula of $\prox_{\lambda g_1}(\bm \alpha)$ is known for $p = 1$ and $2$: 
for $p=1$ 
\begin{align}
  (\prox_{\lambda g_1}(\bm \alpha))_i
  =
  \begin{dcases*}
    \max \set{\alpha_i - \lambda, \alpha^{\mathrm{init}}_i } & if $\alpha_i \geq \alpha^{\mathrm{init}}_i$,\\
    \min \set{\alpha_i  + \lambda, \alpha^{\mathrm{init}}_i } & otherwise, 
  \end{dcases*}
\end{align}
and for $p=2$
\begin{align}
  (\prox_{\lambda g_1}(\bm \alpha))_i
  = \alpha^{\mathrm{init}}_i + \frac{\alpha_i - \alpha^{\mathrm{init}}_i}{2\lambda + 1}. 
\end{align}

In \Cref{line:initialization}, we have to set an upper bound on $\lambda^*$ in \Cref{thm:bisection}.
For $p=1$, we can use $\|\bm \alpha - \bm \alpha^{\mathrm{init}}\|_\infty$ as the upper bound 
because $\prox_{\lambda g_1}(\bm \alpha) = \bm \alpha^{\mathrm{init}}$ holds
for $\lambda = \|\bm \alpha - \bm \alpha^{\mathrm{init}}\|_\infty$.

\begin{algorithm}[t]
  \caption{Projection onto the feasible region $\mathcal C$}
  \label{alg:projection}
  \SetKwInOut{Input}{Input}
  \SetKwInOut{Output}{Output}
  \Input{\ 
    $\bm \alpha \in \R^V$, \ $T > 0$
  }
  \Output{\ $\proj_{\mathcal C}(\bm \alpha)$ with $O(2^{-T})$-error}
  \If{$\proj_{[\bm l, \bm u]}(\bm \alpha) \in \mathcal C$}{
    \Return $\proj_{[\bm l, \bm u]}(\bm \alpha)$\;
  }
  $\lambda_\mathrm{left} \gets 0$, \ $\lambda_\mathrm{right} \gets \text{(An upper bound of $\lambda^*$ in \Cref{thm:bisection})}$\label{line:initialization}\;
  \For{$t = 1,2,\dots,T$}{
    $\lambda_\mathrm{mid} \gets (\lambda_\mathrm{left} + \lambda_\mathrm{right}) / 2$\;
    \If{
      $\proj_{[\bm l, \bm u]}( \prox_{\lambda_{\mathrm{mid}} g_1} (\bm \alpha) ) \in \mathcal C$
      \label{line:check_feasiblity}
    }{
      \tcp{$g_1$ is defined by \cref{eq:def_of_g1}}
      $\lambda_\mathrm{right} \gets \lambda_\mathrm{mid}$\;
    }
    \Else{
      $\lambda_\mathrm{left} \gets \lambda_\mathrm{mid}$\;
    }
  }
  \Return $\proj_{[\bm l, \bm u]}( \prox_{\lambda_{\mathrm{right}} g_1} (\bm \alpha) )$\label{line:return_projected_vec}\;
\end{algorithm}

 \section{Experimental Evaluation}\label{sec:experiments}

In this section, we evaluate the performance of the proposed algorithm in terms of both the quality of solutions and running time.

\subsection{Setup}
Here we explain the experimental setup in details. 
Throughout the experiments, we use $p=1$ and $2$. 

\paragraph{Instances.} 
Table~\ref{tab:instance} lists real-world graphs on which our experiments were conducted. 
As in the description, all graphs consist of social interactions among individuals, implying their suitability for use in our experiments. 

For each graph $G=(V,E)$, we generate an instance of our problem as follows: 
The procedure is almost identical to that of Chan et al.~\cite{Chan+19}; 
the only difference can be found in the determination of the parameter $k\in \mathbb{R}_{\geq 0}$ of the $\ell_p$-norm constraint 
because the problem studied in Chan et al.~\cite{Chan+19} does not have such a constraint. 
Specifically, the innate opinion $s_i$ of each vertex $i\in V$ is selected uniformly at random from $[0,1]$. 
For each $(i,j)\in V\times V$, we pick $w_{ij}=w_{ji}$ uniformly at random from $[0,1]$ if $\{i,j\}\in E$ and set $w_{ij}=w_{ji}=0$ otherwise. 
Then for each $(i,j)\in V\times V$, we set $P_{ij}=w_{ij}/\sum_{k\in V}w_{ik}$. 
The lower bound $l_i$ on the resistance of each vertex $i\in V$ is set to $0.001$ with probability $0.99$ 
and is selected uniformly at random from $[0.001,0.1]$ with probability $0.01$. 
The upper bound $u_i$ is set to $0.999$ with probability $0.99$ 
and is selected uniformly at random from $[0.9,0.999]$ with probability $0.01$. 
Then the initial resistance $\alpha^\text{init}_i$ of each vertex $i\in V$ is selected uniformly at random from $[l_i,u_i]$. 

To observe the performance of the proposed algorithm when changing the strength of the $\ell_p$-norm constraint, 
we use a variety of settings of parameter $k\in \mathbb{R}_{\geq 0}$ as follows: 
First we solve the problem studied in Chan et al.~\cite{Chan+19} (with the instance generated above) 
using their algorithm and obtain an optimal solution, which we denote by $\bm{\alpha}^\text{Chan+}$. 
Then we compute the value of $k'=\|\bm{\alpha}^\text{Chan+}-\bm{\alpha}^\text{init}\|_p$. 
Using this information, we generate an instance for each $k=ck'$ ($c=0.0,0.1,\dots, 1.0$). 
As the value of $c$ increases, the $\ell_p$-norm constraint becomes looser. 
For instance, when $c=0$, $\bm{\alpha}^\text{init}$ must be a unique feasible solution, 
whereas when $c=1$, $\bm{\alpha}^\text{Chan+}$ becomes a trivial optimal solution. 
\begin{table}[t]
\begin{center}
\caption{Real-world graphs used in our experiments.}\label{tab:instance}
\scalebox{0.95}{
\begin{tabular}{lrrl}
\toprule
Name & $|V|$   & $|E|$   &Description \\
\midrule
\textsf{Lesmis          }  &      77     &      254     &Co-appearance\\%
\textsf{Jazz            }  &     198     &     2,742    &Social network\\%
\textsf{Email           }  &    1,133    &     5,451    &Email communication\\%
\textsf{ca-GrQc         }  &    4,158    &    13,422    &Collaboration \\%
\textsf{ca-HepPh        }  &   11,204    &   117,619    &Collaboration \\%
\textsf{ca-AstroPh      }  &   17,903    &   196,972    &Collaboration \\%
\textsf{ca-CondMat      }  &   21,363    &    91,286    &Collaboration \\%
\textsf{email-Enron     }  &   33,696    &   180,811    &Email communication\\%
\textsf{soc-Epinions1   }  &   75,877    &   405,739    &Social network\\%
\textsf{soc-Slashdot0902}  &   82,168    &   504,230    &Social network\\%
\textsf{com-DBLP        }  &  317,080    &  1,049,866   &Co-authorship \\%
\textsf{com-Youtube     }  &  1,134,890  &  2,987,624   &Social network\\%
\bottomrule
\end{tabular}
}
\end{center}
\end{table}

\paragraph{Baseline methods.}

We employ the following three baselines, all of which are based on a simple greedy strategy starting from an initial solution specified. 

The first one, which we call \textsf{Gradient\_$\bm{\alpha}^\mathrm{Chan+}$},  starts from the solution $\bm{\alpha}^\text{Chan+}$. 
As $\bm{\alpha}^\text{Chan+}$ is not a feasible solution in most cases (due to the above instance generation procedure), 
the algorithm aims to reach a feasible solution with the increase of the objective function value as small as possible. 
To this end, the algorithm iteratively selects a vertex $v$ 
with minimum $\left|\frac{\partial f(\bm{\alpha})}{\partial \alpha_v}\right|$, where $\bm{\alpha}$ is a current solution, among the vertices that have not yet been selected 
and shifts its resistance to the initial resistance $\alpha^\text{init}_v$. 
The complete procedure is described in Algorithm~\ref{alg:baseline_Chan}. 
\begin{algorithm}[t]
\caption{\textsf{Gradient\_$\bm{\alpha}^\mathrm{Chan+}$}}\label{alg:baseline_Chan}
\SetKwInOut{Input}{Input}
\SetKwInOut{Output}{Output}
\Input{\ $V$, $\bm{s}\in [0,1]^{V}$, $P\in [0,1]^{V\times V}$, $\bm{l}, \bm{\alpha}^\mathrm{init}, \bm{u}\in (0,1)^{V}$ ($\bm{0}<\bm{l}\leq \bm{\alpha}^\mathrm{init}\leq \bm{u}\leq \bm{1}$), and $k\in \mathbb{R}_{\geq 0}$}
\Output{\ $\bm{\alpha}\in (0,1]^V$}
$\bm{\alpha} \leftarrow \bm{\alpha}^\text{Chan+}$\;
$V_\mathrm{cand}\leftarrow V$\;
\While{$\|\bm{\alpha}-\bm{\alpha}^\mathrm{init}\|_p>k$}{
Find $v\in \argmin\left\{\left|\frac{\partial f(\bm{\alpha})}{\partial \alpha_v}\right| : v\in V_\mathrm{cand}\right\}$\;
$\alpha_v\leftarrow \alpha^\mathrm{init}_v$\;
$V_\mathrm{cand}\leftarrow V_\mathrm{cand}\setminus \{v\}$\;
}
\Return $\bm{\alpha}$\;
\end{algorithm}

The second one, which we refer to as \textsf{Gradient\_$\bm{\alpha}^\mathrm{init}$}, starts from the initial resistance vector $\bm{\alpha}^\text{init}$, 
which is always a feasible solution of our problem. 
The algorithm aims to decrease the objective function value as much as possible without violating the $\ell_p$-norm constraint. 
To this end, the algorithm iteratively selects a vertex $v$ 
with maximum $\left|\frac{\partial f(\bm{\alpha})}{\partial \alpha_v}\right|$ among the vertices that have not yet been selected; 
we set its resistance to $l_v$ if $\frac{\partial f(\bm{\alpha})}{\partial \alpha_v}$ is nonnegative and to $u_v$ otherwise. 
The pseudo-code is given in Algorithm~\ref{alg:baseline_alpha-init}. 

The above two baselines are expected to perform well in practice because they effectively use the information of the gradient of the objective function. 
However, they are not applicable to large-sized instances because they require the computation of the gradient $O(|V|^2)$ times in the worst case. 
Therefore, we employ a faster method, which we call \textsf{Column-sum\_$\bm{\alpha}^\mathrm{Chan+}$}. 
This method is a variant of \textsf{Gradient\_$\bm{\alpha}^\mathrm{Chan+}$}. 
The only difference can be found in the way to choose a vertex for which we change the resistance value. 
Specifically, in this method, we first sort the elements of $V$ in increasing order of the column-sum of interaction matrix $P$, and go through the list. 
This method comes from the intuition that the change of resistance values of vertices with a small column-sum in $P$ 
may not have a significant effect on the objective function value. 
The complete procedure is described in Algorithm~\ref{alg:baseline_degree}. 
Note that it is not trivial to give an $\bm{\alpha}^\text{init}$ counterpart of \textsf{Column-sum\_$\bm{\alpha}^\mathrm{Chan+}$}. 
In fact, we are not sure how to change the value of $\alpha^\text{init}_v$ (e.g., to $l_v$ or $u_v$) because we no longer have the information of the gradient.

\begin{algorithm}[t]
\caption{\textsf{Gradient\_$\bm{\alpha}^\mathrm{init}$}}\label{alg:baseline_alpha-init}
\SetKwInOut{Input}{Input}
\SetKwInOut{Output}{Output}
\Input{\ $V$, $\bm{s}\in [0,1]^{V}$, $P\in [0,1]^{V\times V}$, $\bm{l}, \bm{\alpha}^\mathrm{init}, \bm{u}\in (0,1)^{V}$ ($\bm{0}<\bm{l}\leq \bm{\alpha}^\mathrm{init}\leq \bm{u}\leq \bm{1}$), and $k\in \mathbb{R}_{\geq 0}$}
\Output{\ $\bm{\alpha}\in (0,1]^V$}
$\bm{\alpha} \leftarrow \bm{\alpha}^\mathrm{init}$\;
$V_\mathrm{cand}\leftarrow V$\;
\While{$V_\mathrm{cand}\neq \emptyset$}{
$\bm{\alpha}_\mathrm{prev}\leftarrow \bm{\alpha}$\;
Find $v\in \argmax\left\{\left|\frac{\partial f(\bm{\alpha})}{\partial \alpha_v}\right| : v\in V_\mathrm{cand}\right\}$\;
\textbf{if} $\frac{\partial f(\bm{\alpha})}{\partial \alpha_v}\geq 0$ \textbf{then} $\alpha_v\leftarrow l_v$\;
\textbf{else} $\alpha_v\leftarrow u_v$\;
\textbf{if} $\|\bm{\alpha}-\bm{\alpha}^\mathrm{init}\|_p>k$ \textbf{then} \textbf{return} $\bm{\alpha_\text{prev}}$\;
$V_\mathrm{cand}\leftarrow V_\mathrm{cand}\setminus \{v\}$\;
}
\Return $\bm{\alpha}$\;
\end{algorithm}

\begin{algorithm}[t]
\caption{\textsf{Column-sum\_$\bm{\alpha}^\mathrm{Chan+}$}}\label{alg:baseline_degree}
\SetKwInOut{Input}{Input}
\SetKwInOut{Output}{Output}
\Input{\ $V$, $\bm{s}\in [0,1]^{V}$, $P\in [0,1]^{V\times V}$, $\bm{l}, \bm{\alpha}^\mathrm{init}, \bm{u}\in (0,1)^{V}$ ($\bm{0}<\bm{l}\leq \bm{\alpha}^\mathrm{init}\leq \bm{u}\leq \bm{1}$), and $k\in \mathbb{R}_{\geq 0}$}
\Output{\ $\bm{\alpha}\in (0,1]^V$}
$\bm{\alpha} \leftarrow \bm{\alpha}^\text{Chan+}$\;
Sort the elements of $V$ so that $(v_1,v_2,\dots, v_{|V|})$ with $\sum_{v\in V}P_{v,v_1}\leq \sum_{v\in V}P_{v,v_2}\leq \cdots \leq \sum_{v\in V}P_{v,v_{|V|}}$\;
\For{$i=1,2,\dots, |V|$}{
  \textbf{if} $\|\bm{\alpha}-\bm{\alpha}^\mathrm{init}\|_p\leq k$ \textbf{then} \textbf{break}\;
  $\alpha_i\leftarrow \alpha^\mathrm{init}_i$\;
}
\Return $\bm{\alpha}$\;
\end{algorithm}

\begin{figure*}[t]
\scalebox{0.963}{
\begin{subfloat}[ca-HepPh]{
\includegraphics[width=0.292\textwidth]{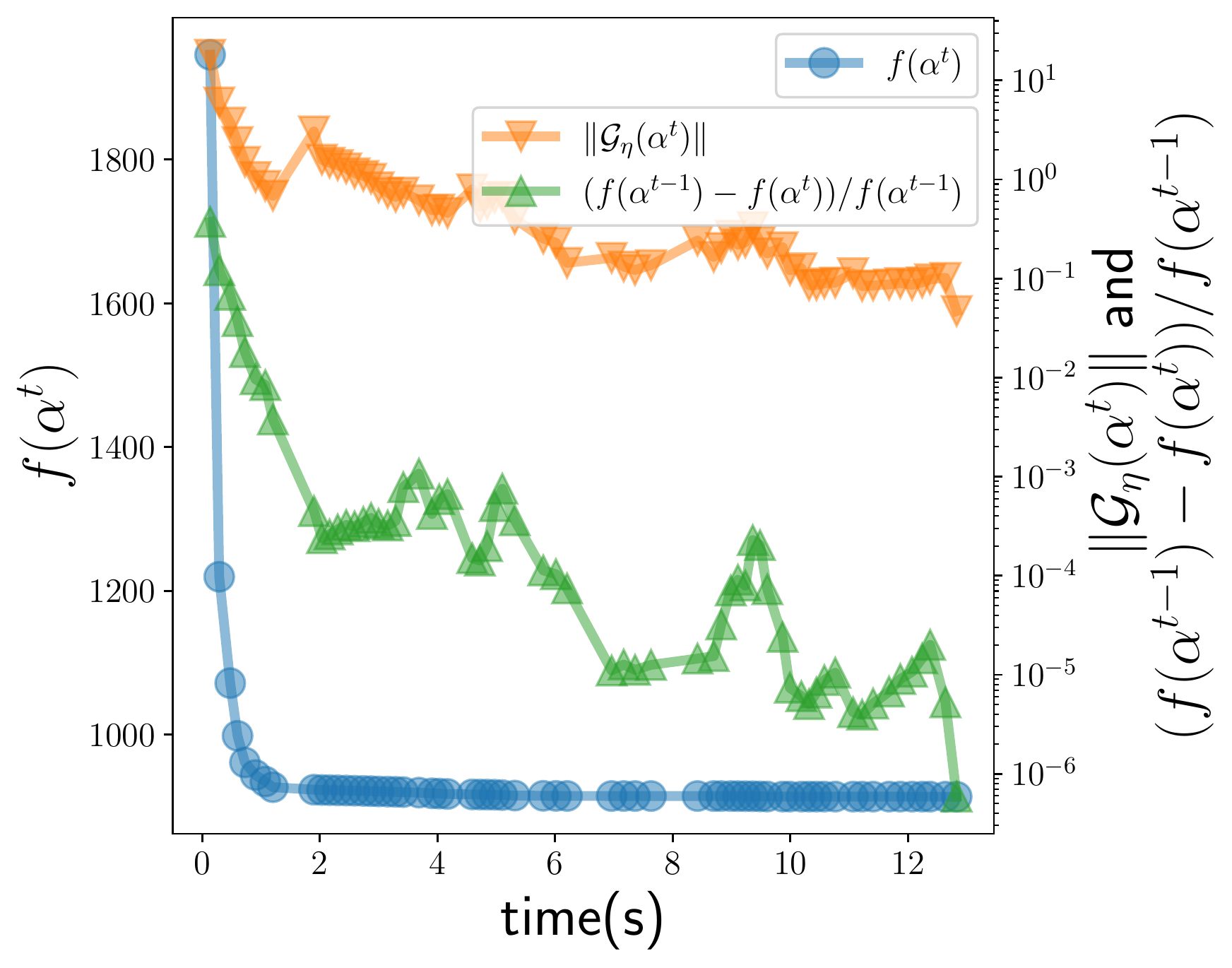}
}
\end{subfloat}
\begin{subfloat}[soc-Epinions1]{
\includegraphics[width=0.296\textwidth]{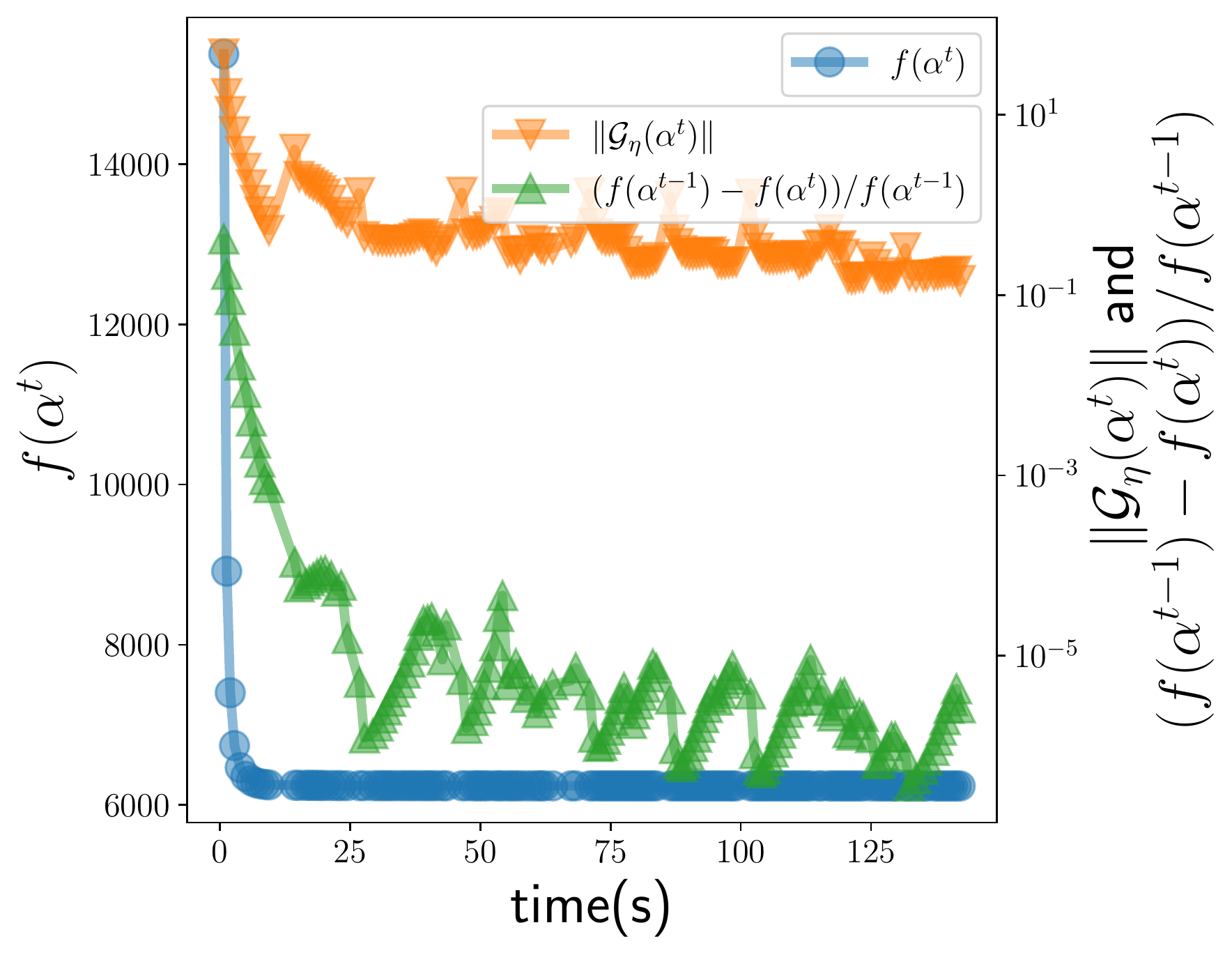}
}
\end{subfloat}
\begin{subfloat}[com-Youtube]{
\includegraphics[width=0.300\textwidth]{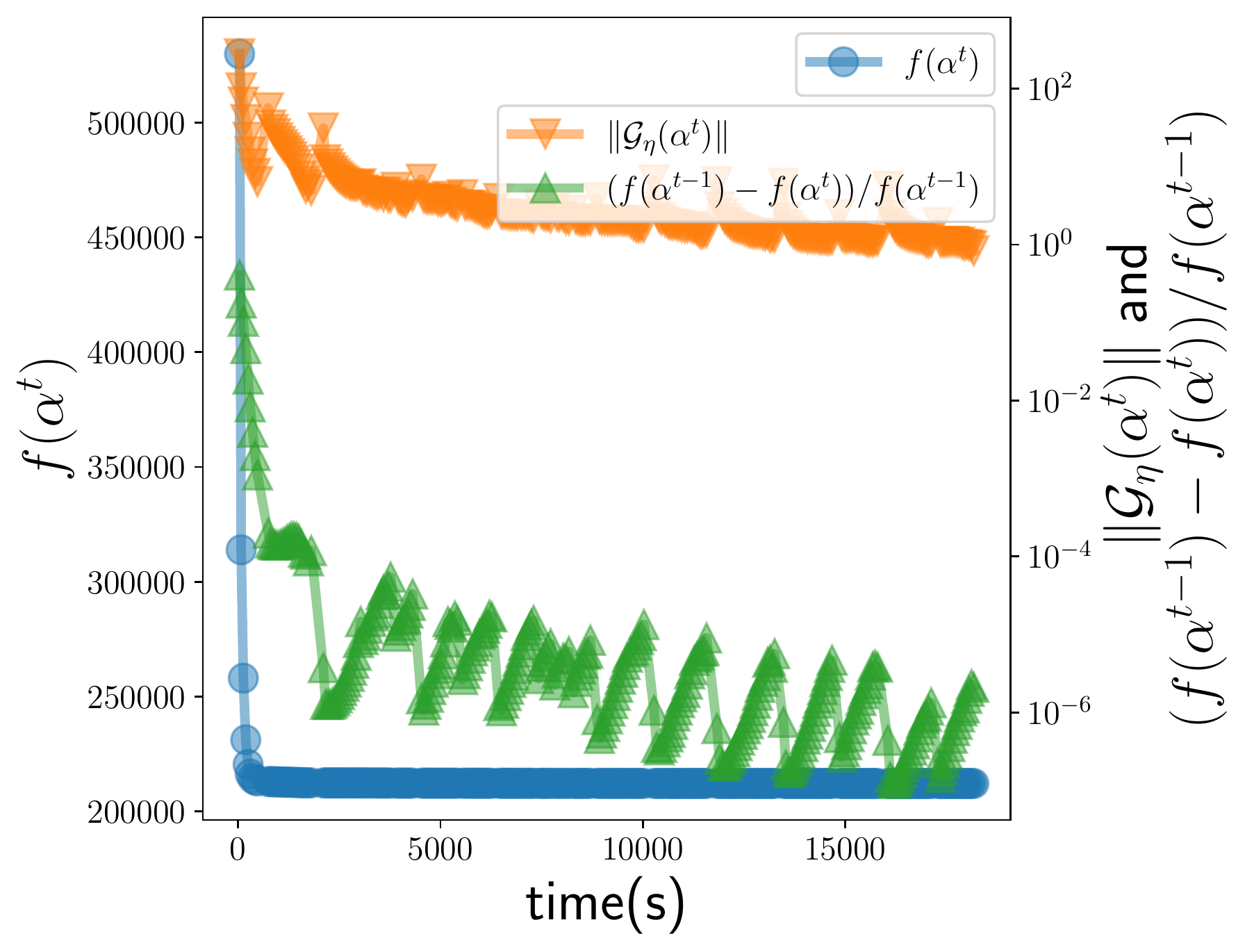}
}
\end{subfloat}
}
\caption{The behavior of our algorithm using a very conservative stopping condition ($\epsilon=10^{-3}\sqrt{|V|}$).}\label{fig:exp_pre}
\end{figure*}

\paragraph{Parameter settings.}
Our algorithm has parameters $\eta_0, \eta_\text{inc}, \epsilon, T>0$ and $\gamma <1$. 
Throughout the experiments, we set $\eta_0=1.0$, $\eta_\text{inc}=1.25$, $T=200$, and $\gamma= 0.5$. 
The remaining parameter $\epsilon$ appears in the stopping condition $\|\bm{\alpha}^t-\bm{\alpha}^{t-1}\|_2\leq \eta \epsilon$. 
To determine it appropriately, we conducted a preliminary experiments, 
by setting $\epsilon$ to a very small (i.e., conservative) value, 
say $\epsilon=10^{-3}\sqrt{|V|}$. 

The results are depicted in Figure~\ref{fig:exp_pre}. 
Due to space limitations, we present the results only for the graphs \textsf{ca-HepPh}, \textsf{soc-Epinions1}, and \textsf{com-Youtube}, 
with $c=0.5$ and the initial solution $\bm{\alpha}^\text{Chan+}$, 
but the trend is similar for the other graphs with the other parameters $c$ and the other initial solution $\bm{\alpha}^\text{init}$. 
Each point corresponds to an iteration of the for-loop of the algorithm. 

As can be seen, the objective function value steeply decreases and then keeps almost the same value 
until the condition $\|\bm{\alpha}^t-\bm{\alpha}^{t-1}\|_2\leq \eta \epsilon$ is met. 
In our experiments, we are interested in the performance of the algorithm in terms of the objective function value 
rather than whether the solution converges. 
Therefore, in our main experiments, we employ the stopping condition using the relative change of the objective function value. 
Specifically, we use the condition $(f(\bm{\alpha}^{t-1})-f(\bm{\alpha}^t))/f(\bm{\alpha}^{t-1})\leq 10^{-3}$. 

The way to choose an initial (feasible) solution in the proposed algorithm is arbitrary. 
In our experiments, we employ the following two solutions: 
the solution generated by projecting $\bm{\alpha}^\text{Chan+}$ onto $\mathcal{C}$ (i.e., $\text{proj}_\mathcal{C}(\bm{\alpha}^\text{Chan+})$) 
and the initial resistance vector $\bm{\alpha}^\text{init}$.

\begin{figure*}
\begin{subfloat}[Lesmis]{
\includegraphics[width=0.22\textwidth]{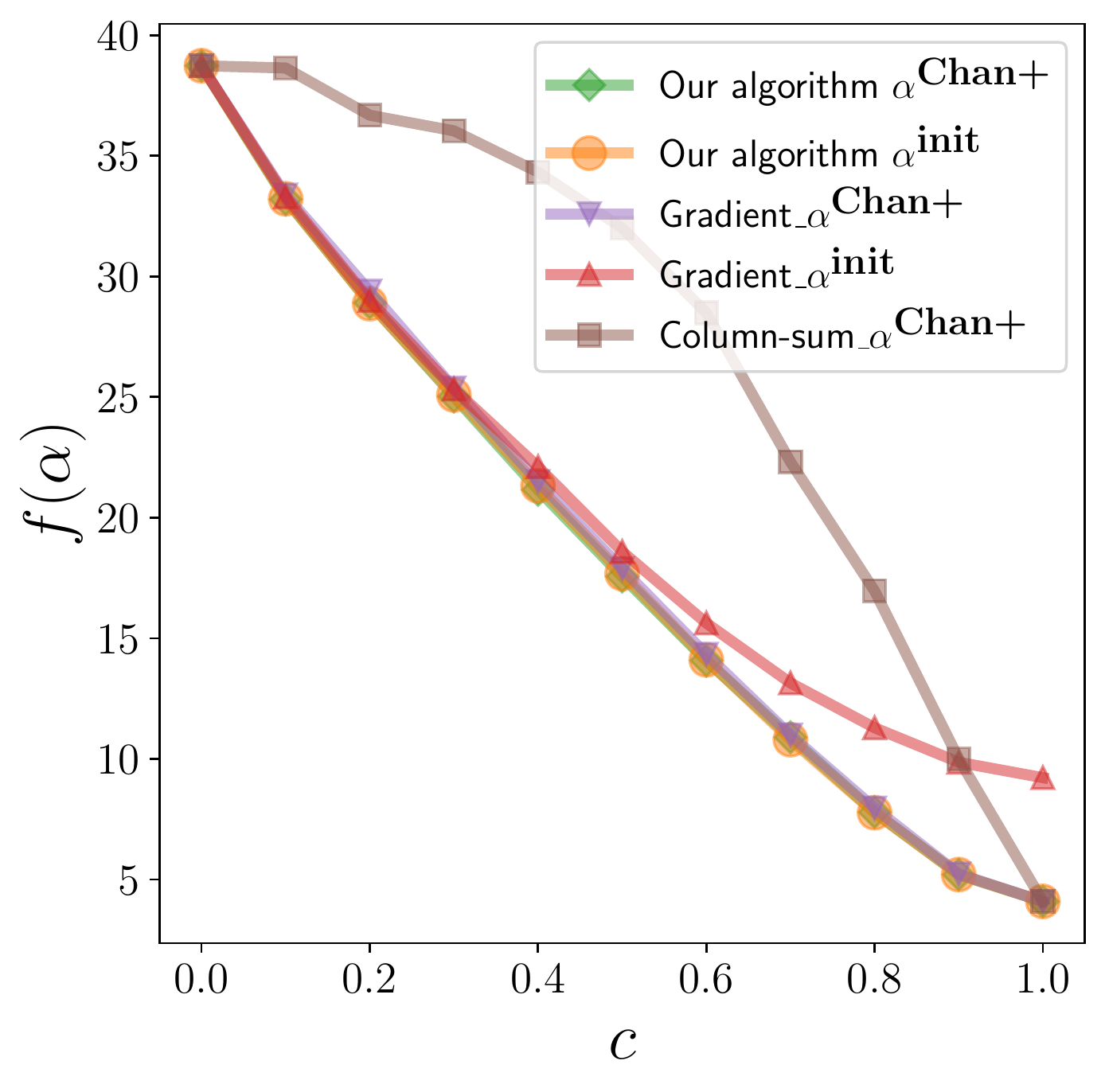}
\includegraphics[width=0.22\textwidth]{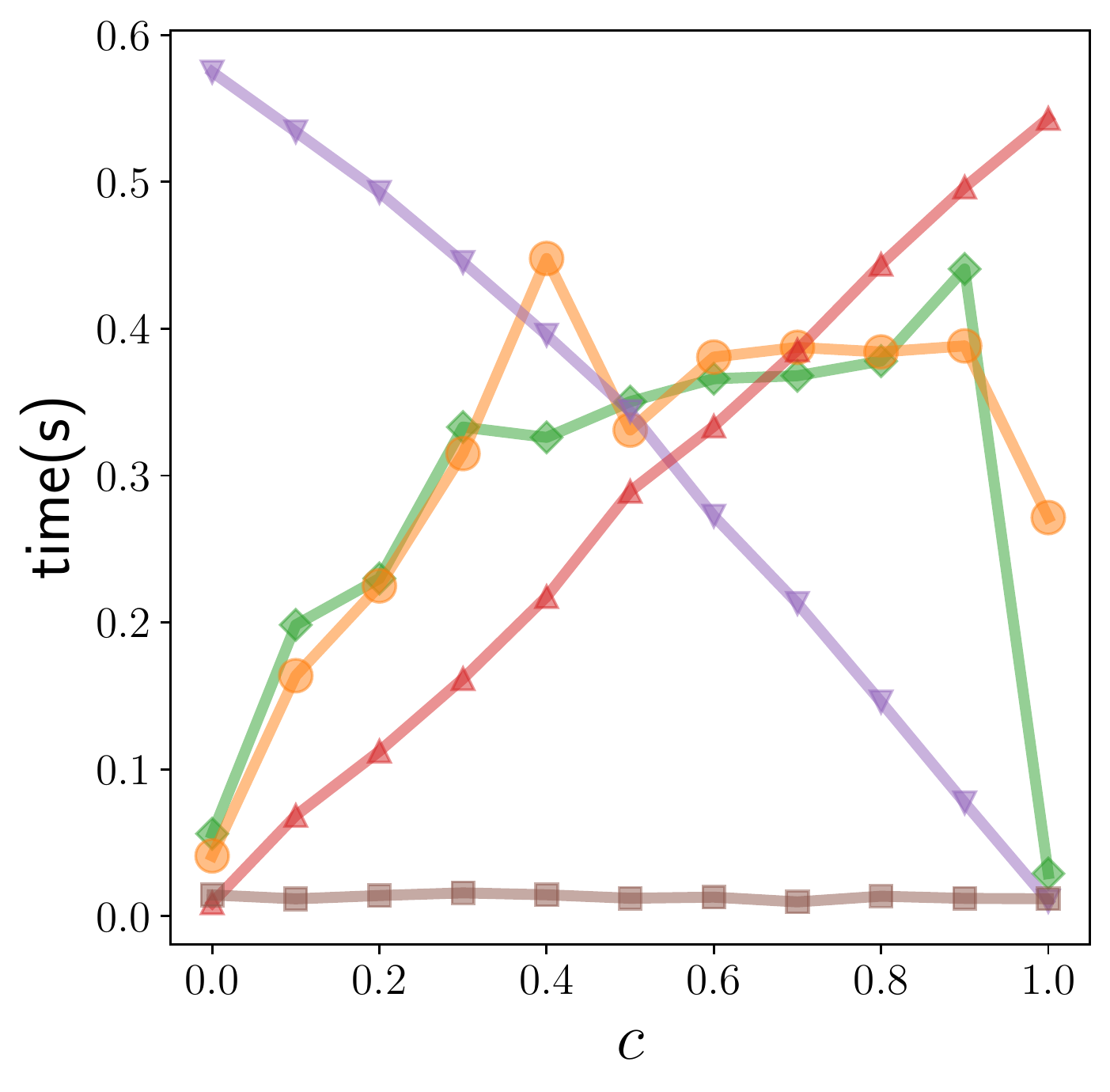}
}
\end{subfloat}
\begin{subfloat}[Jazz]{
\includegraphics[width=0.22\textwidth]{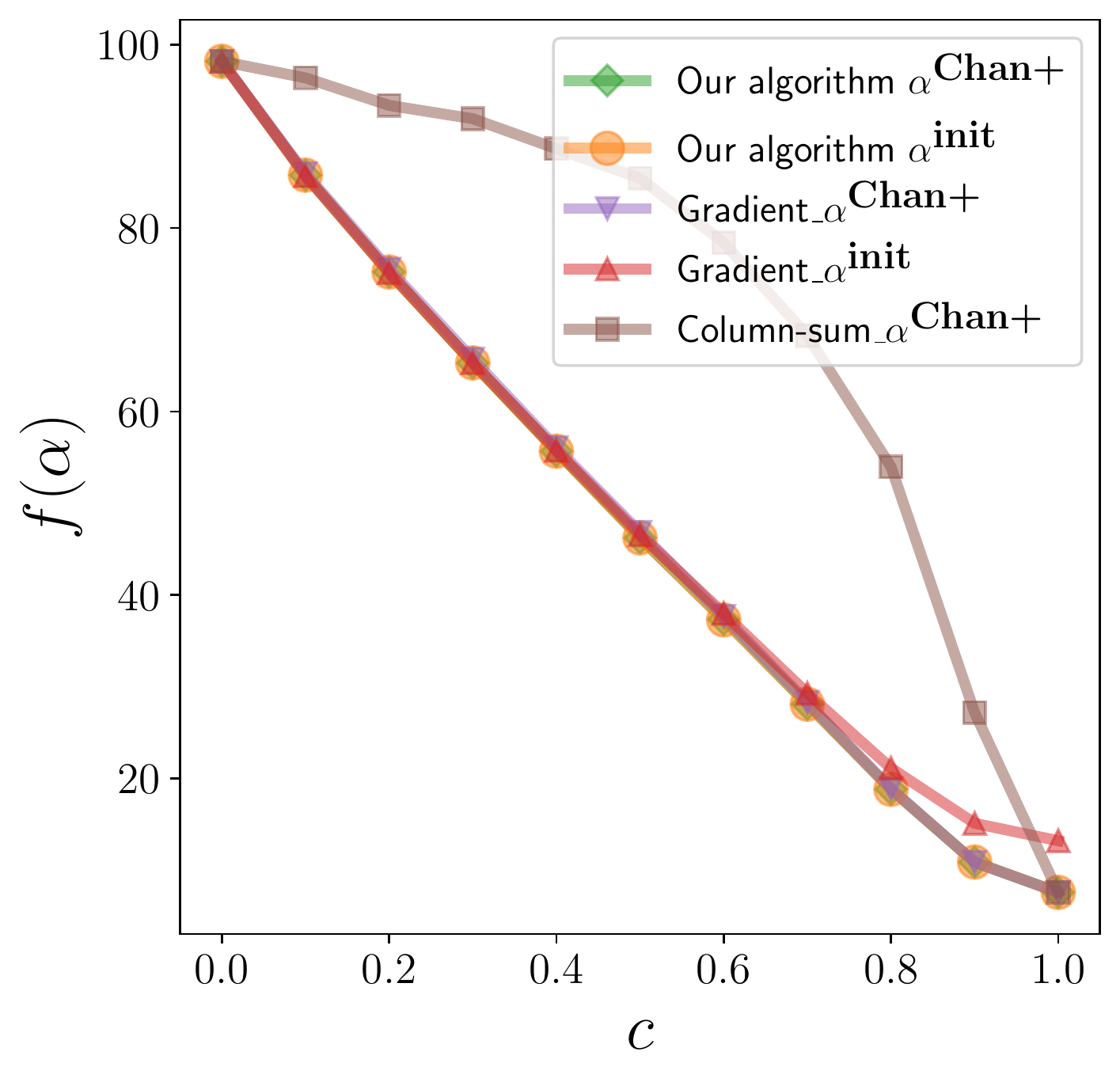}
\includegraphics[width=0.22\textwidth]{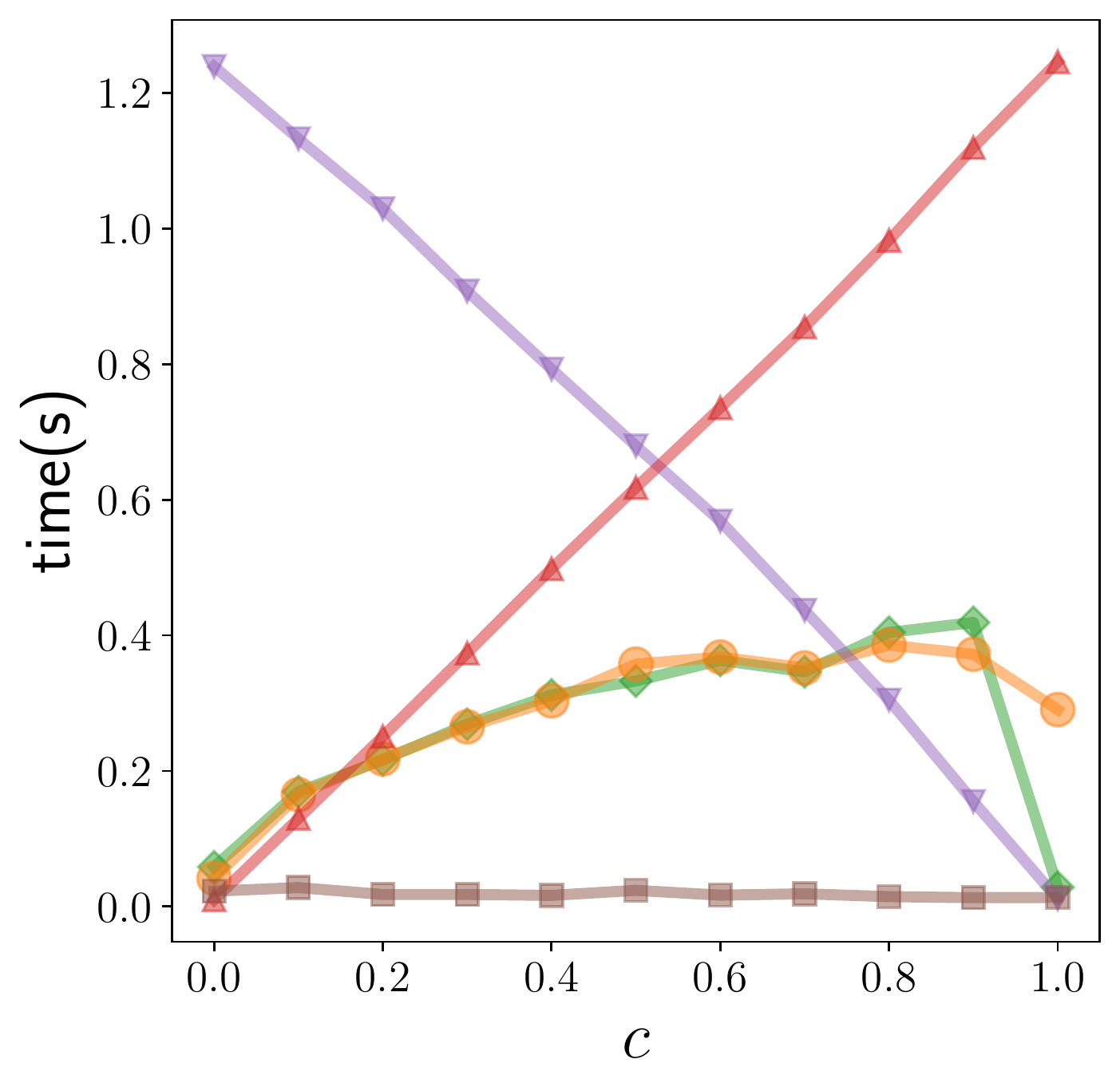}
}
\end{subfloat}
\begin{subfloat}[Email]{
\includegraphics[width=0.22\textwidth]{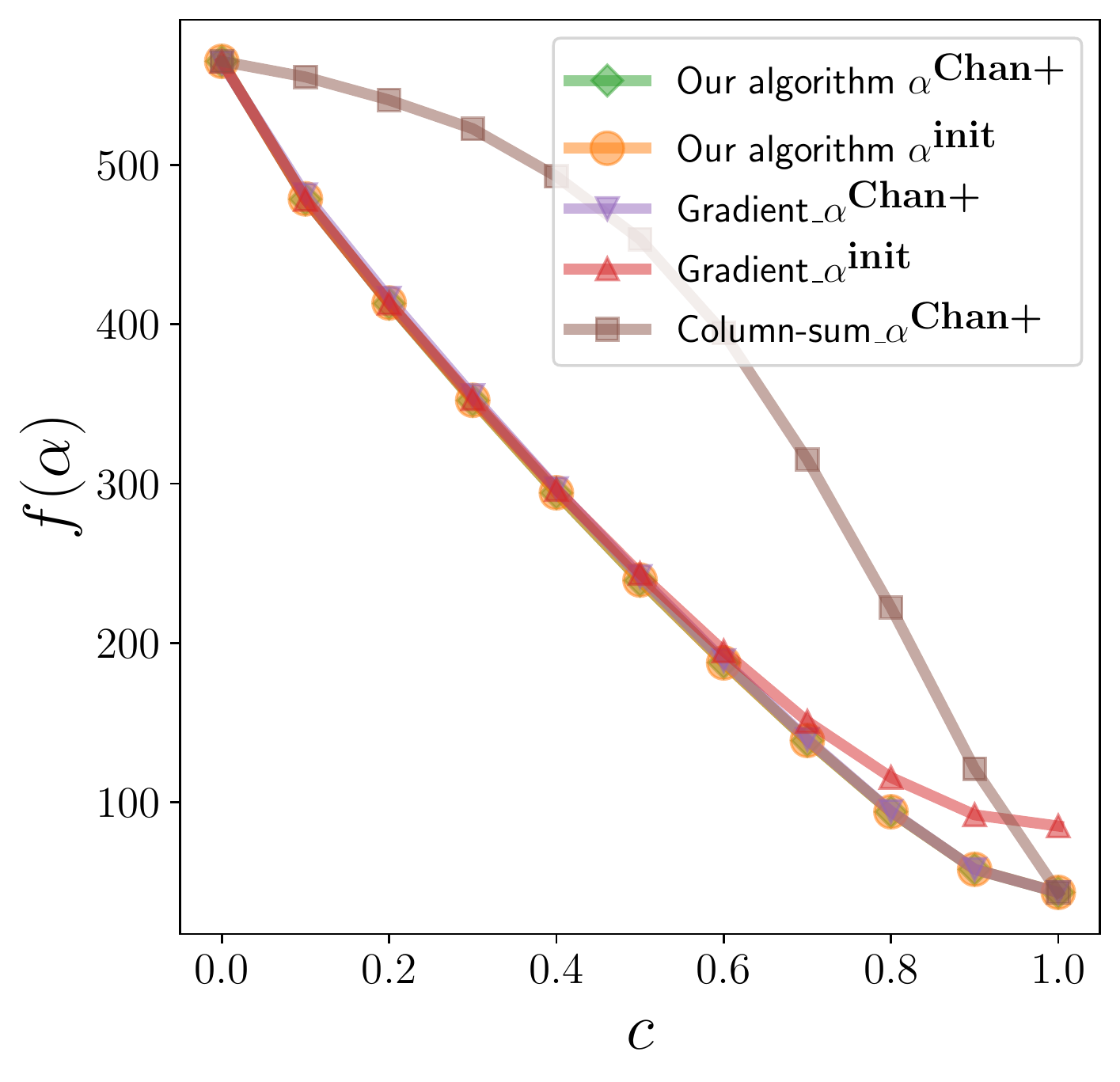}
\includegraphics[width=0.215\textwidth]{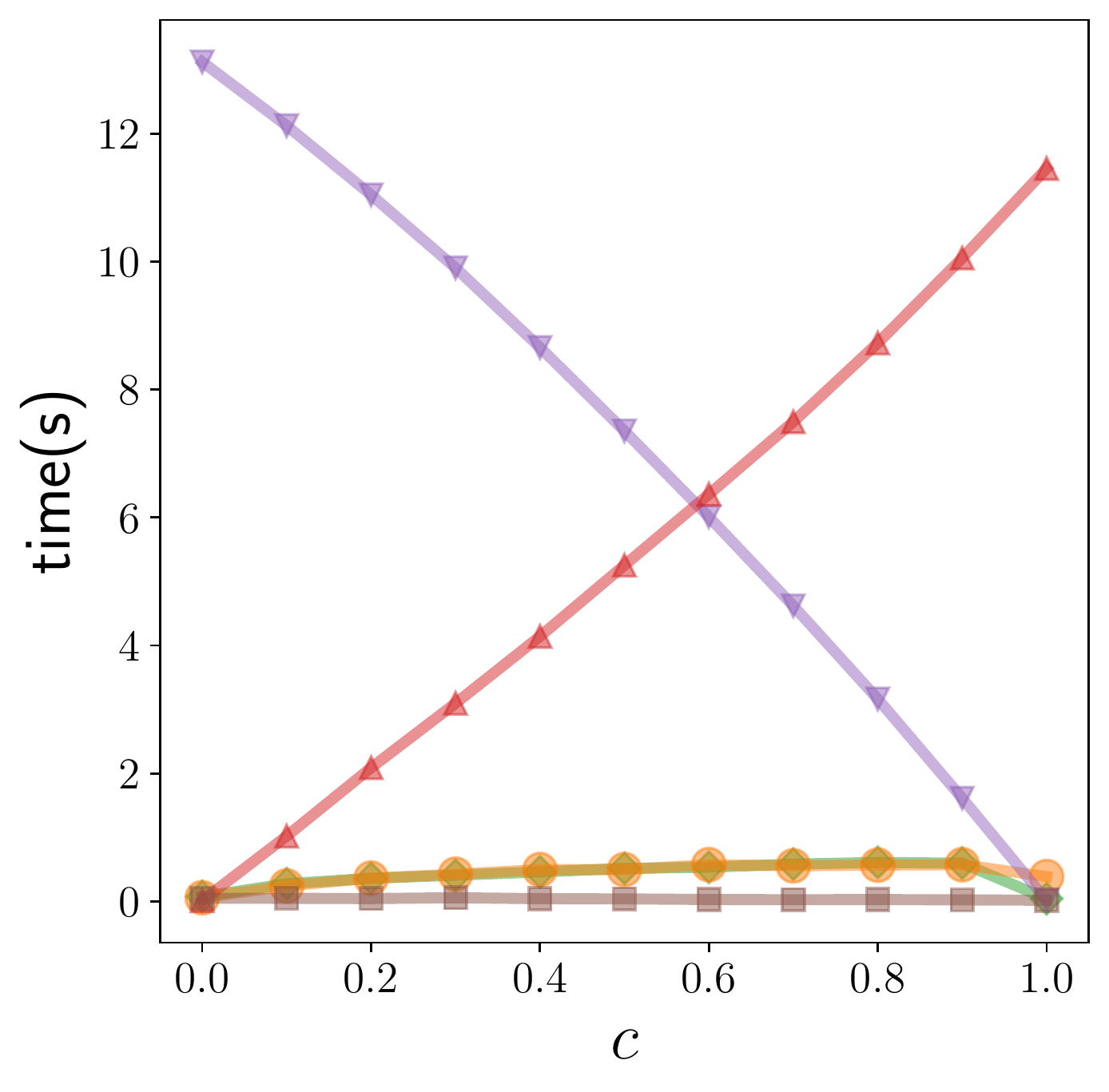}
}
\end{subfloat}
\begin{subfloat}[ca-GrQc]{
\includegraphics[width=0.22\textwidth]{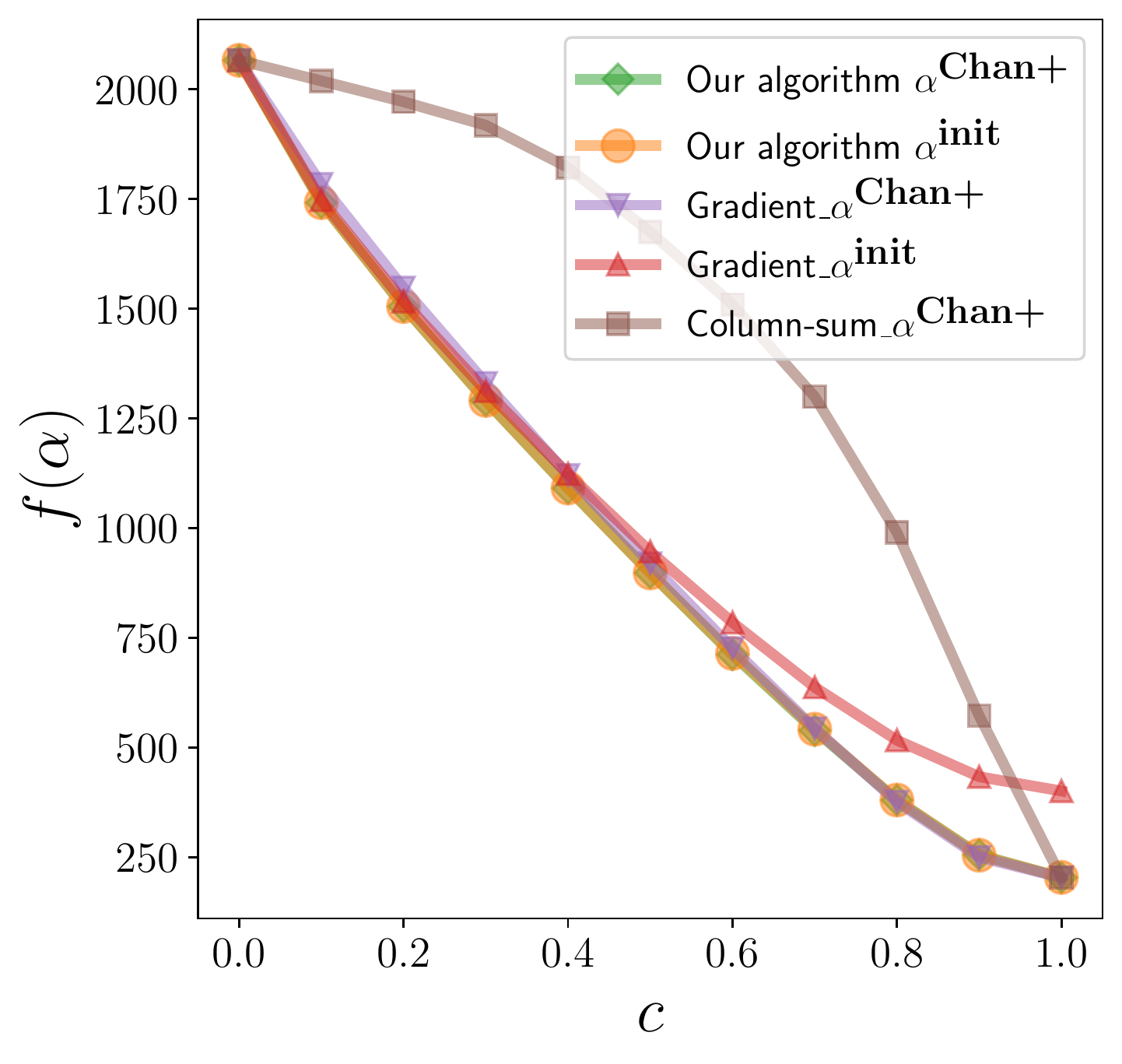}
\includegraphics[width=0.215\textwidth]{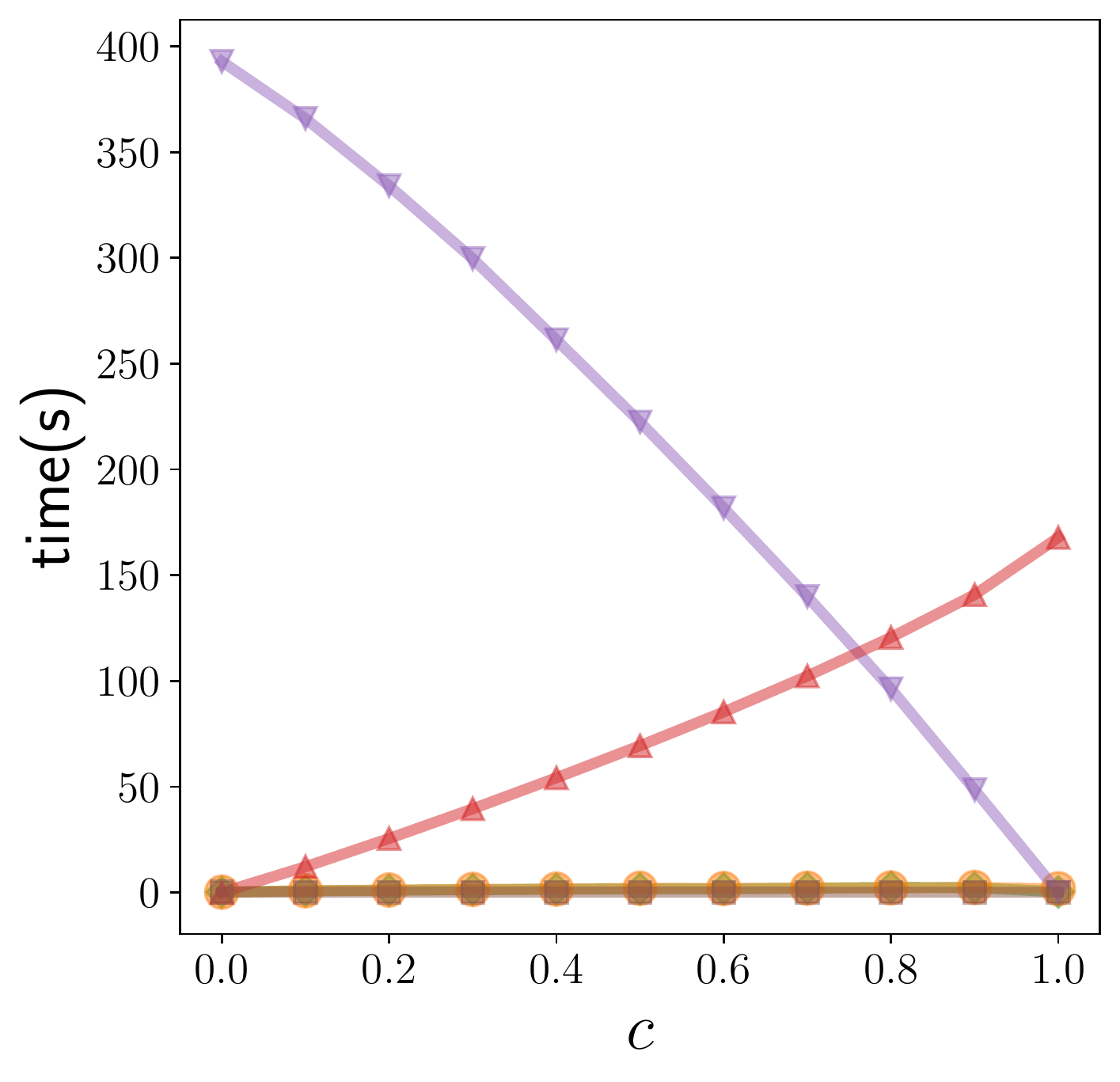}
}
\end{subfloat}
\begin{subfloat}[ca-HepPh]{
\includegraphics[width=0.22\textwidth]{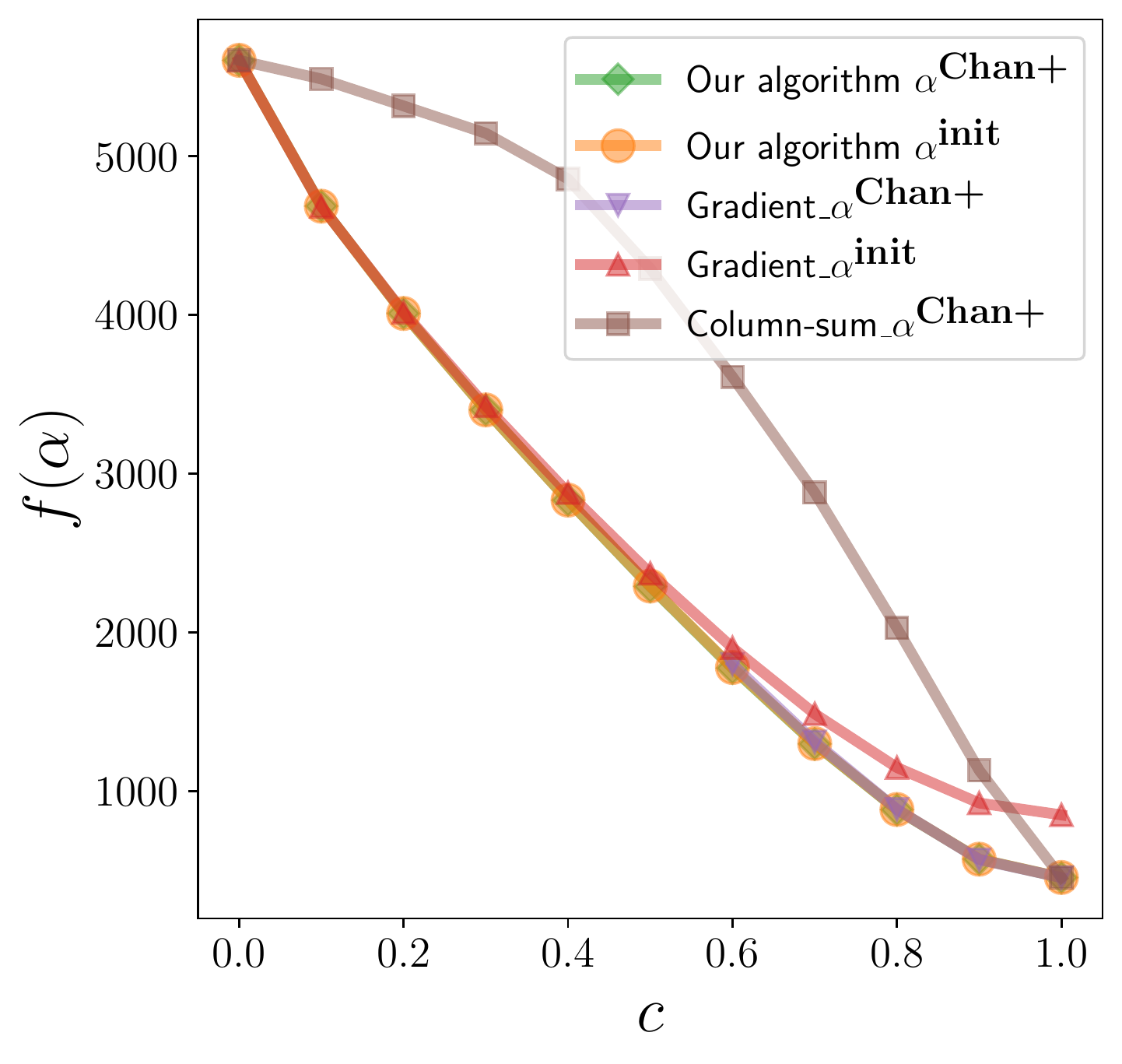}
\includegraphics[width=0.22\textwidth]{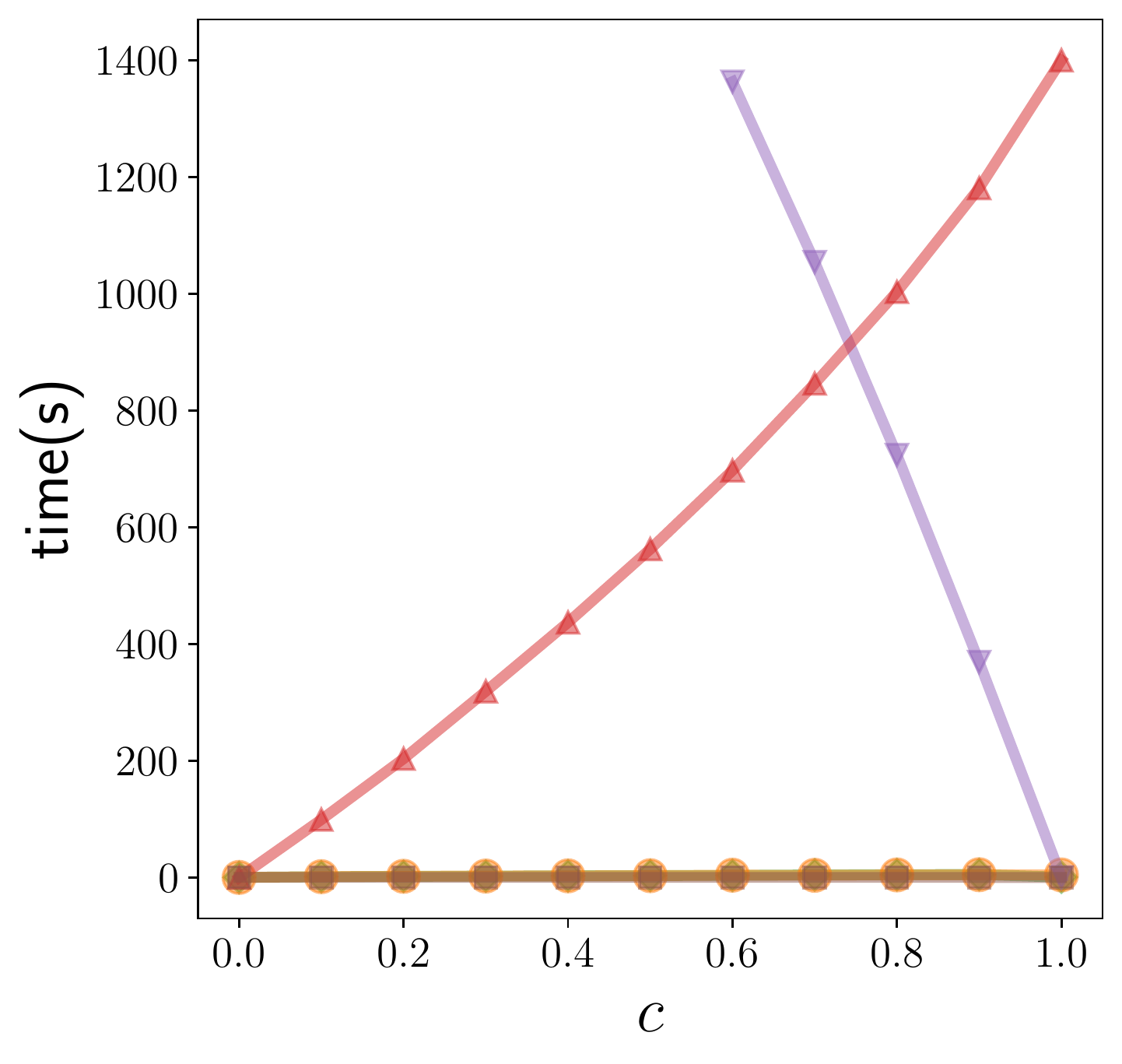}
}
\end{subfloat}
\begin{subfloat}[ca-AstroPh]{
\includegraphics[width=0.22\textwidth]{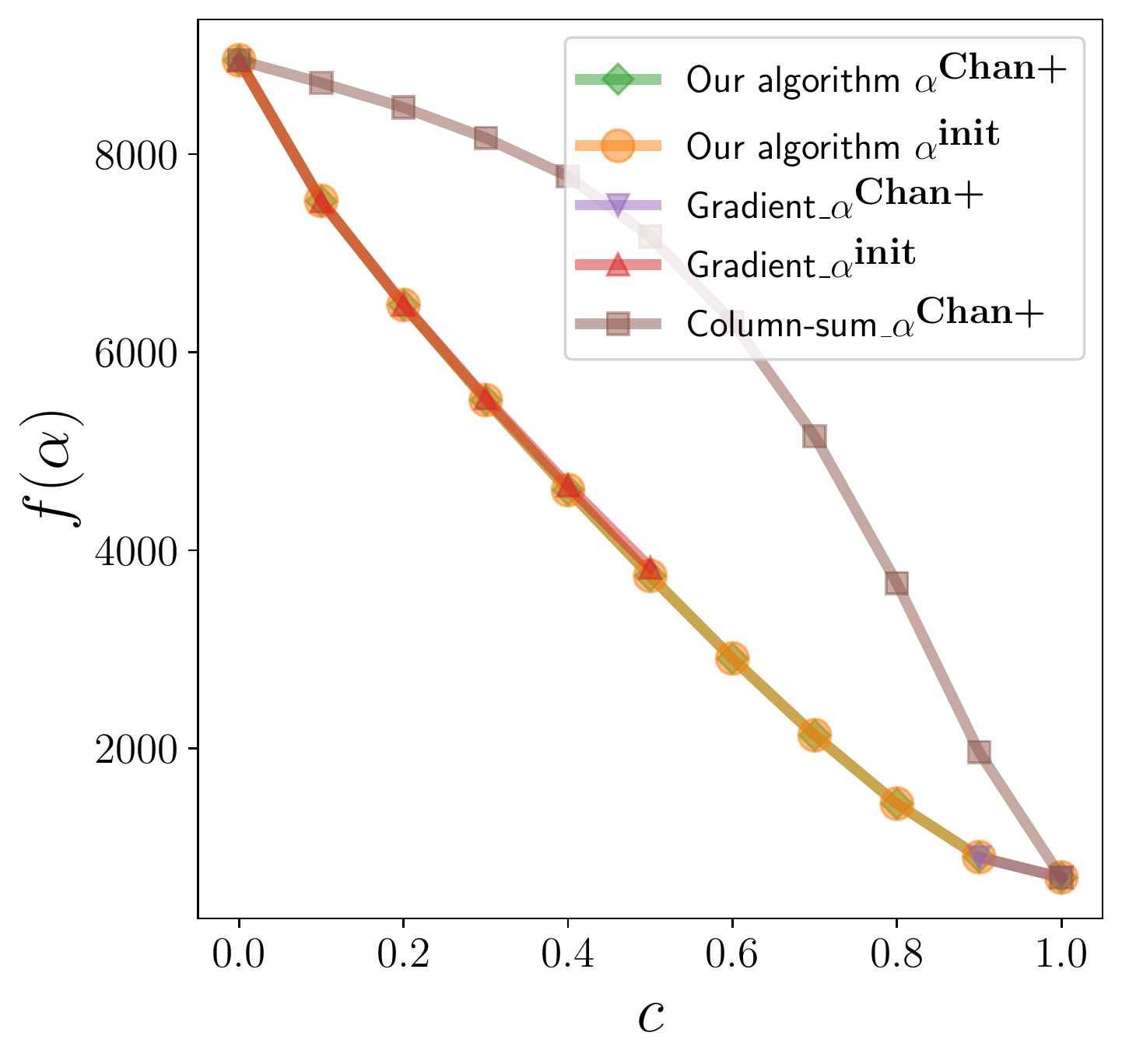}
\includegraphics[width=0.22\textwidth]{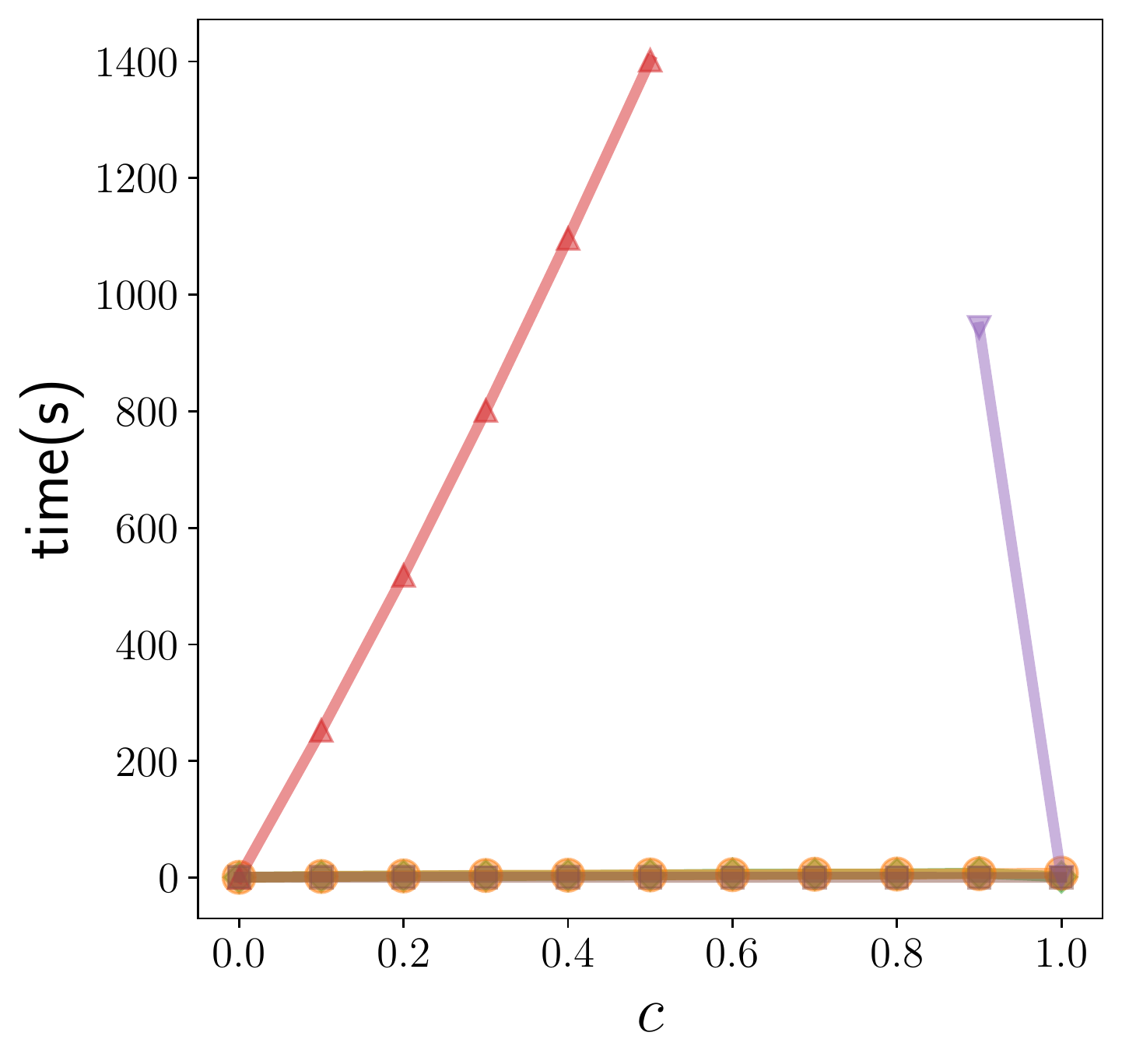}
}
\end{subfloat}
\begin{subfloat}[ca-CondMat]{
\includegraphics[width=0.22\textwidth]{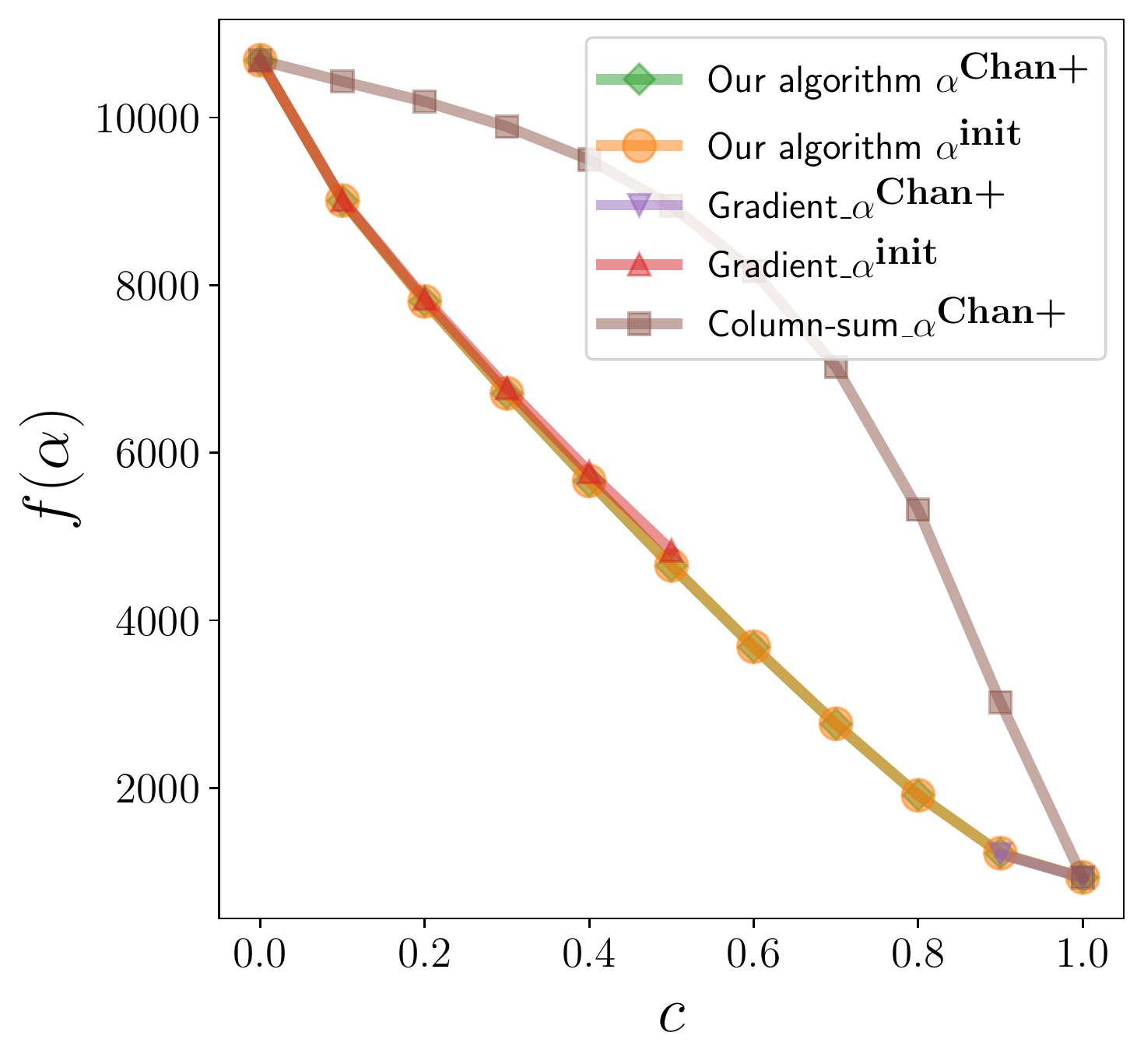}
\includegraphics[width=0.215\textwidth]{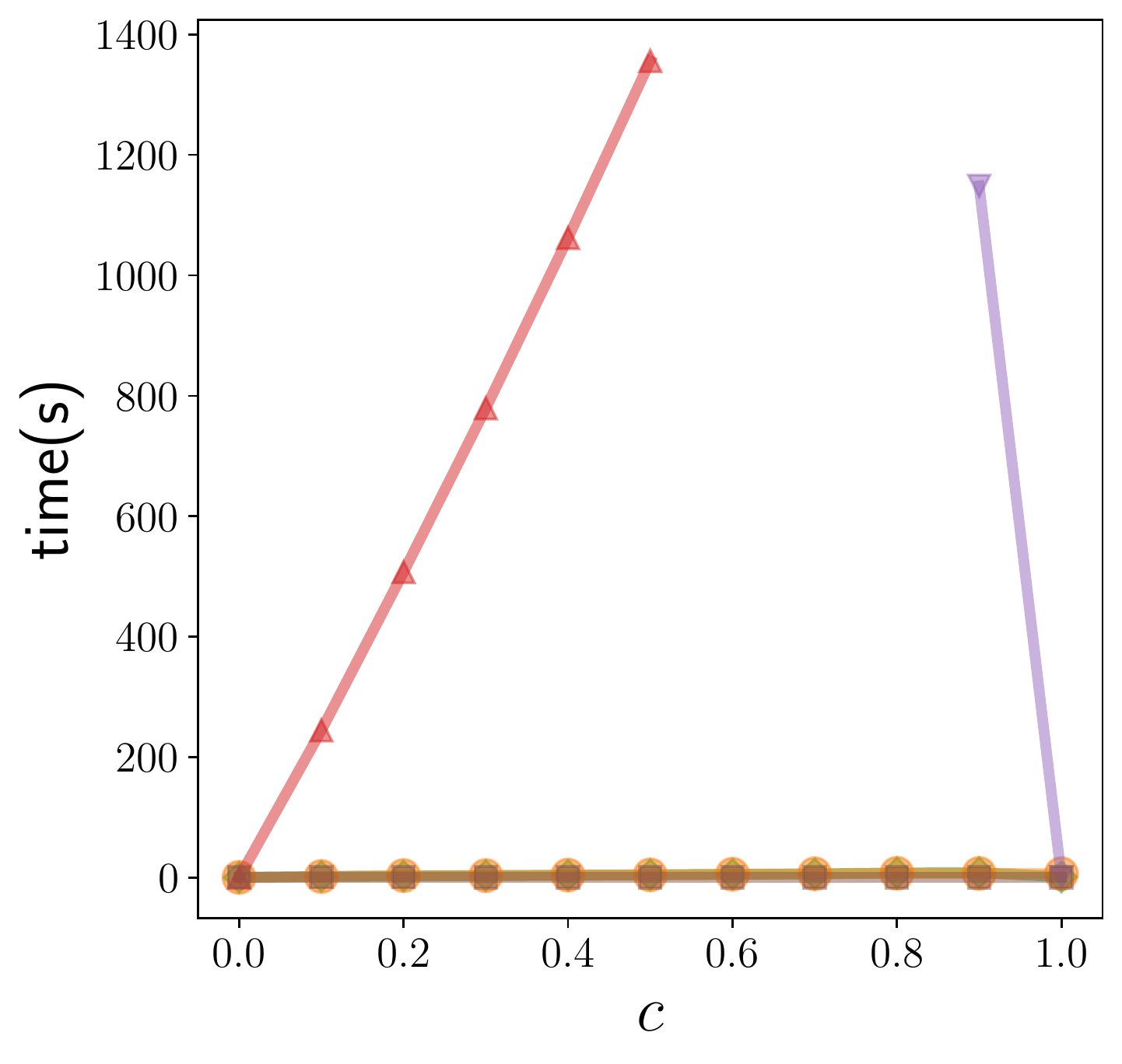}
}
\end{subfloat}
\begin{subfloat}[email-Enron]{
\includegraphics[width=0.22\textwidth]{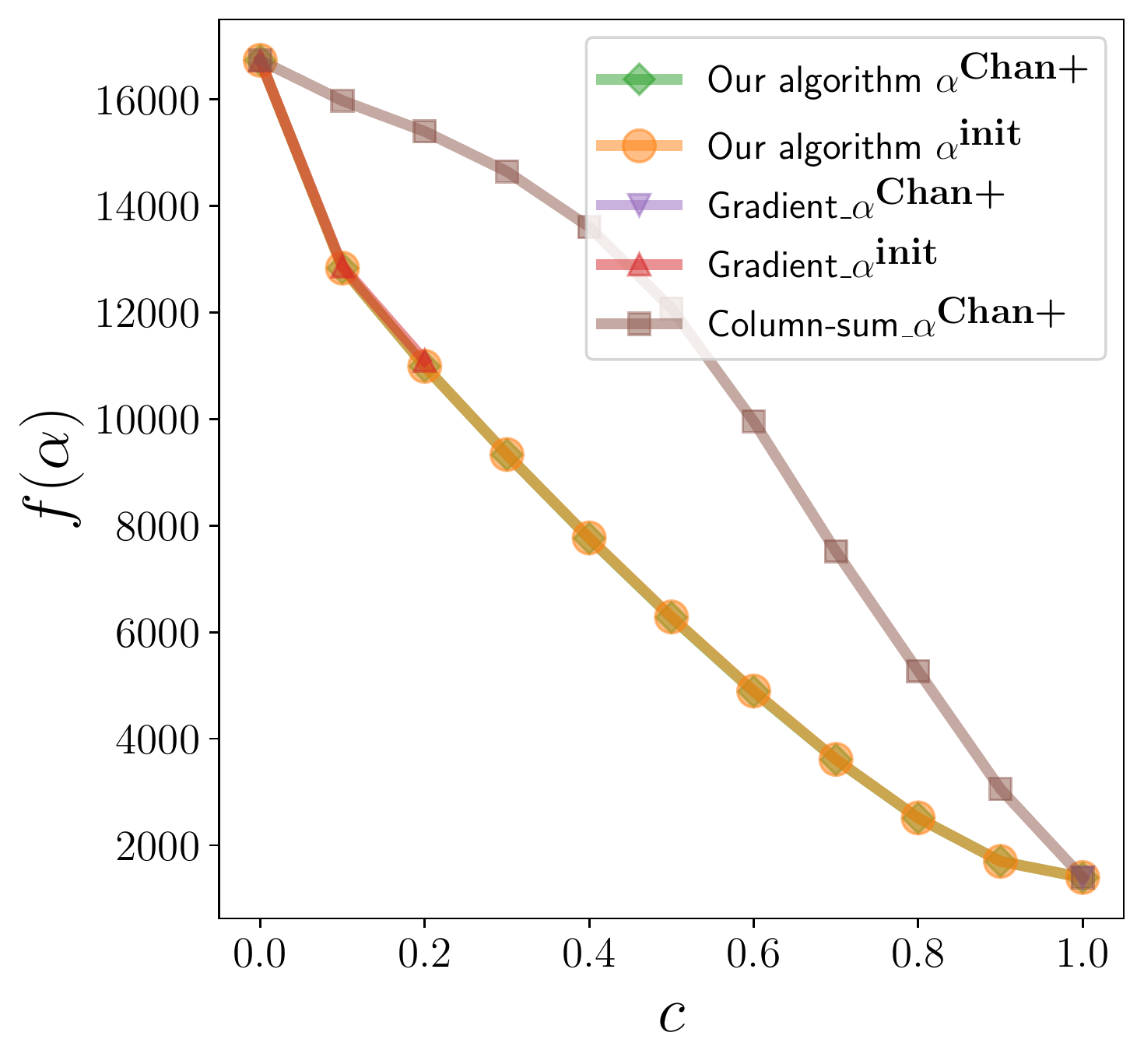}
\includegraphics[width=0.215\textwidth]{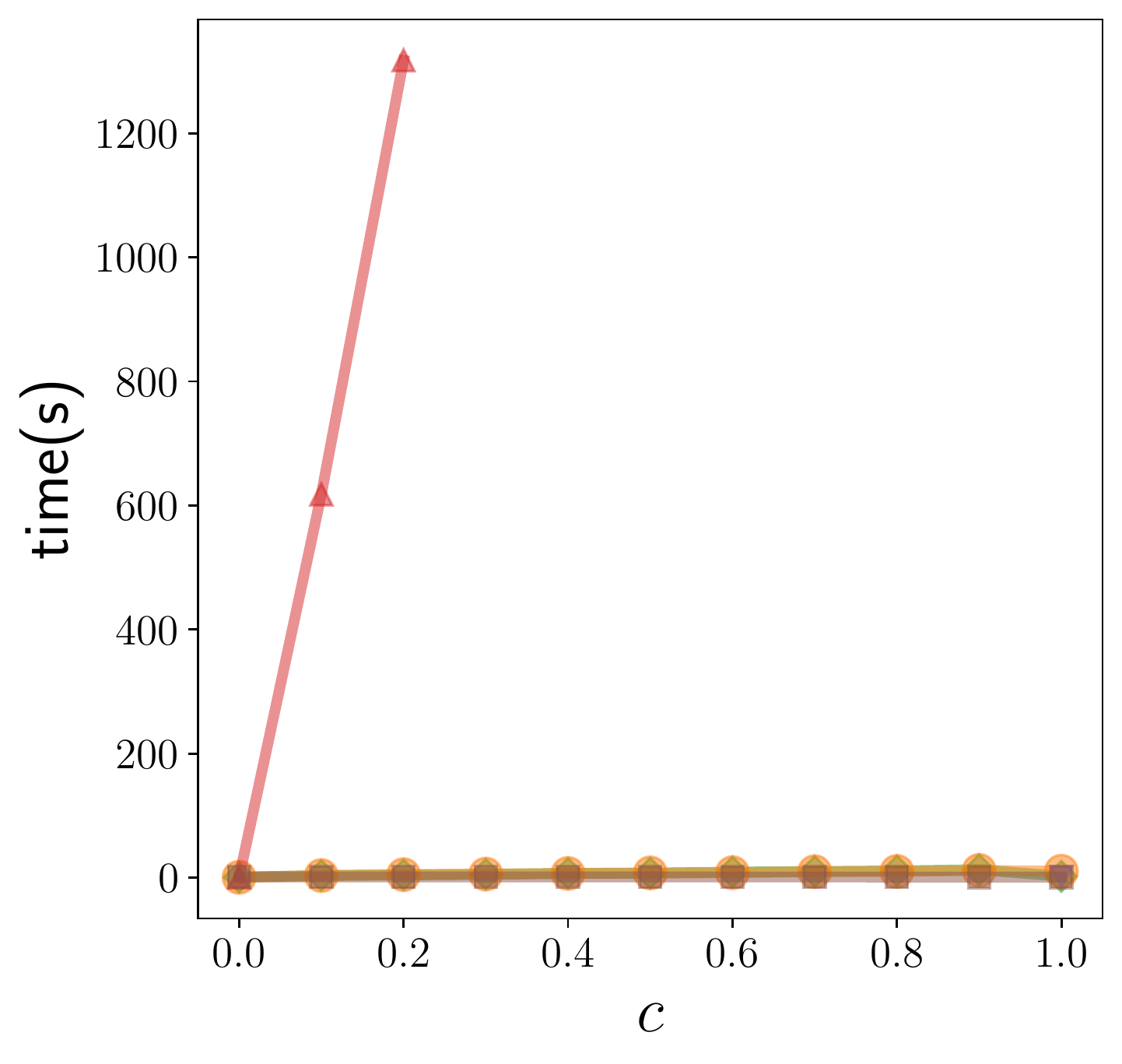}
}
\end{subfloat}
\caption{Results with $p=1$ for the first eight graphs.}\label{fig:result_1}
\end{figure*}

\begin{figure*}
\begin{subfloat}[soc-Epinions1]{
\includegraphics[width=0.22\textwidth]{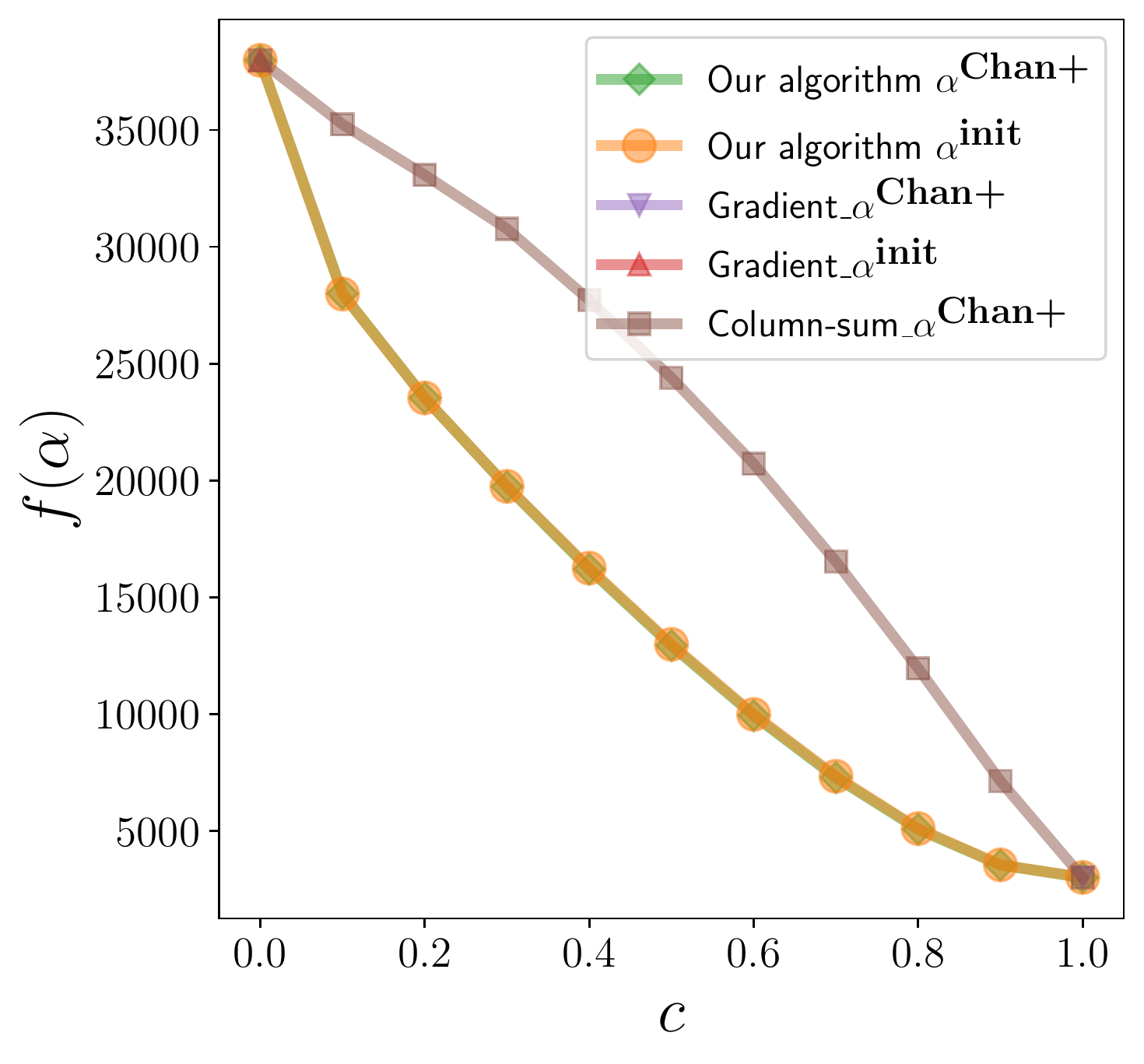}
\includegraphics[width=0.208\textwidth]{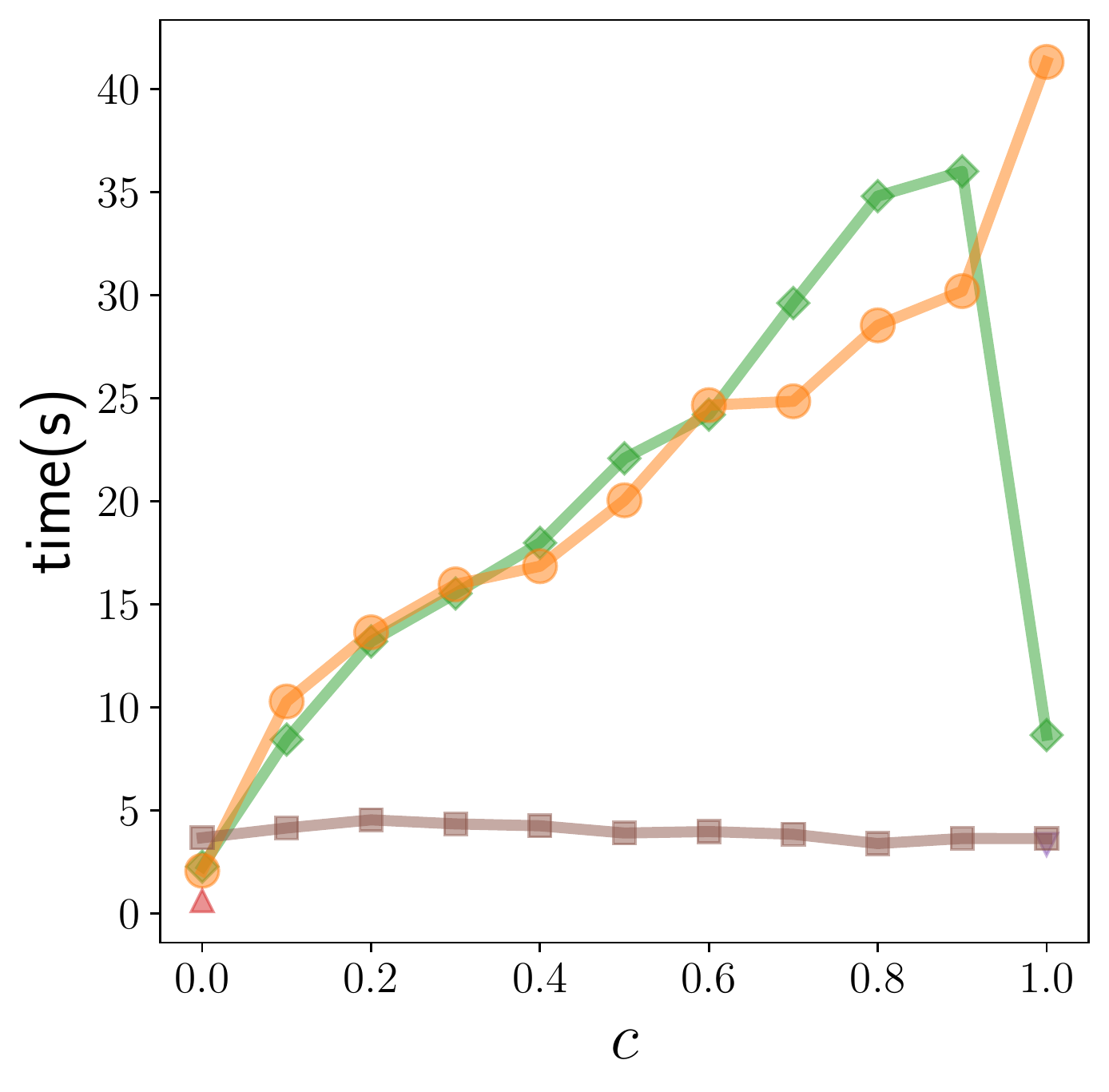}
}
\end{subfloat}
\begin{subfloat}[soc-Slashdot0902]{
\includegraphics[width=0.22\textwidth]{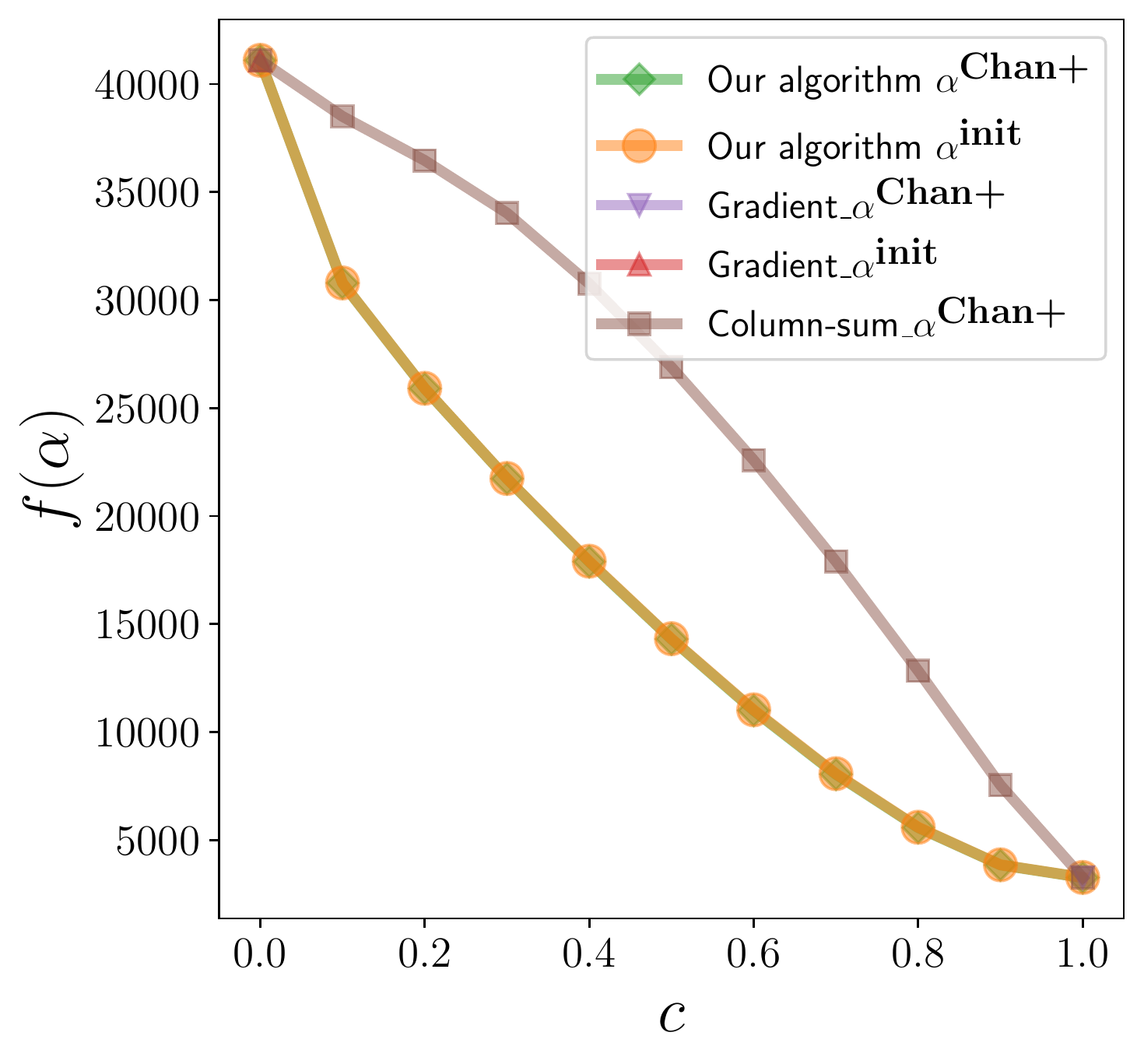}
\includegraphics[width=0.208\textwidth]{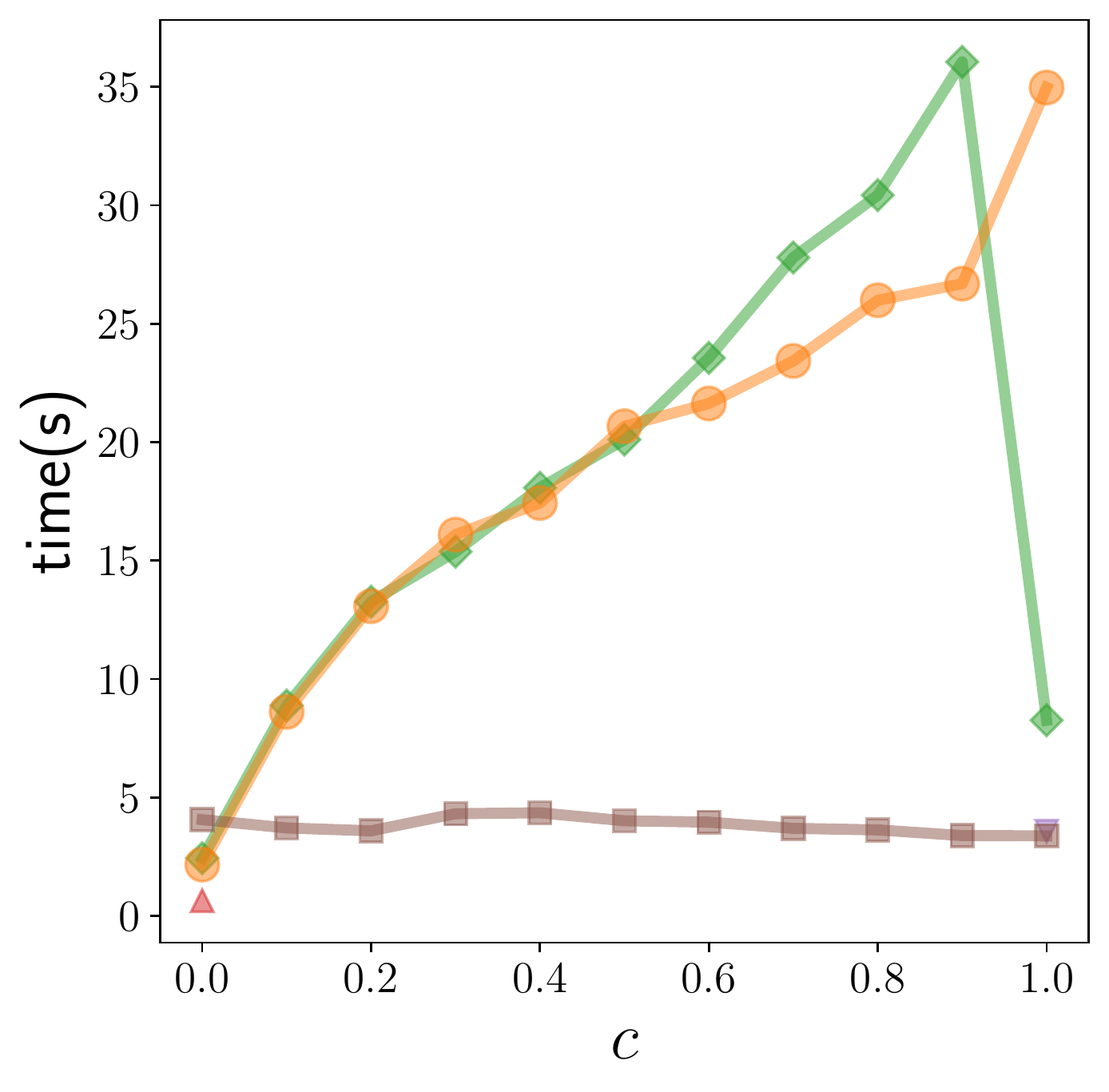}
}
\end{subfloat}
\begin{subfloat}[com-DBLP]{
\includegraphics[width=0.22\textwidth]{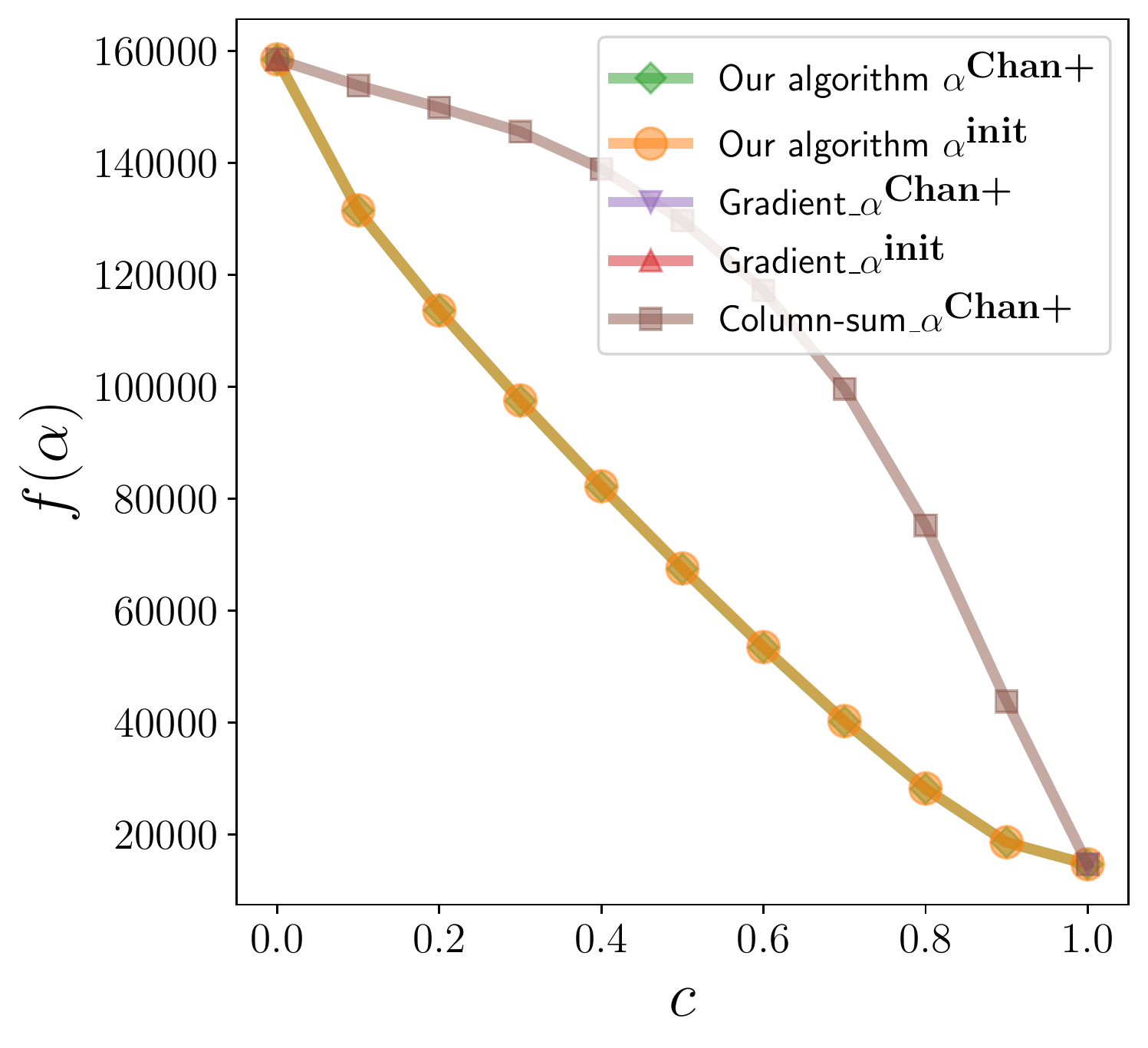}
\includegraphics[width=0.208\textwidth]{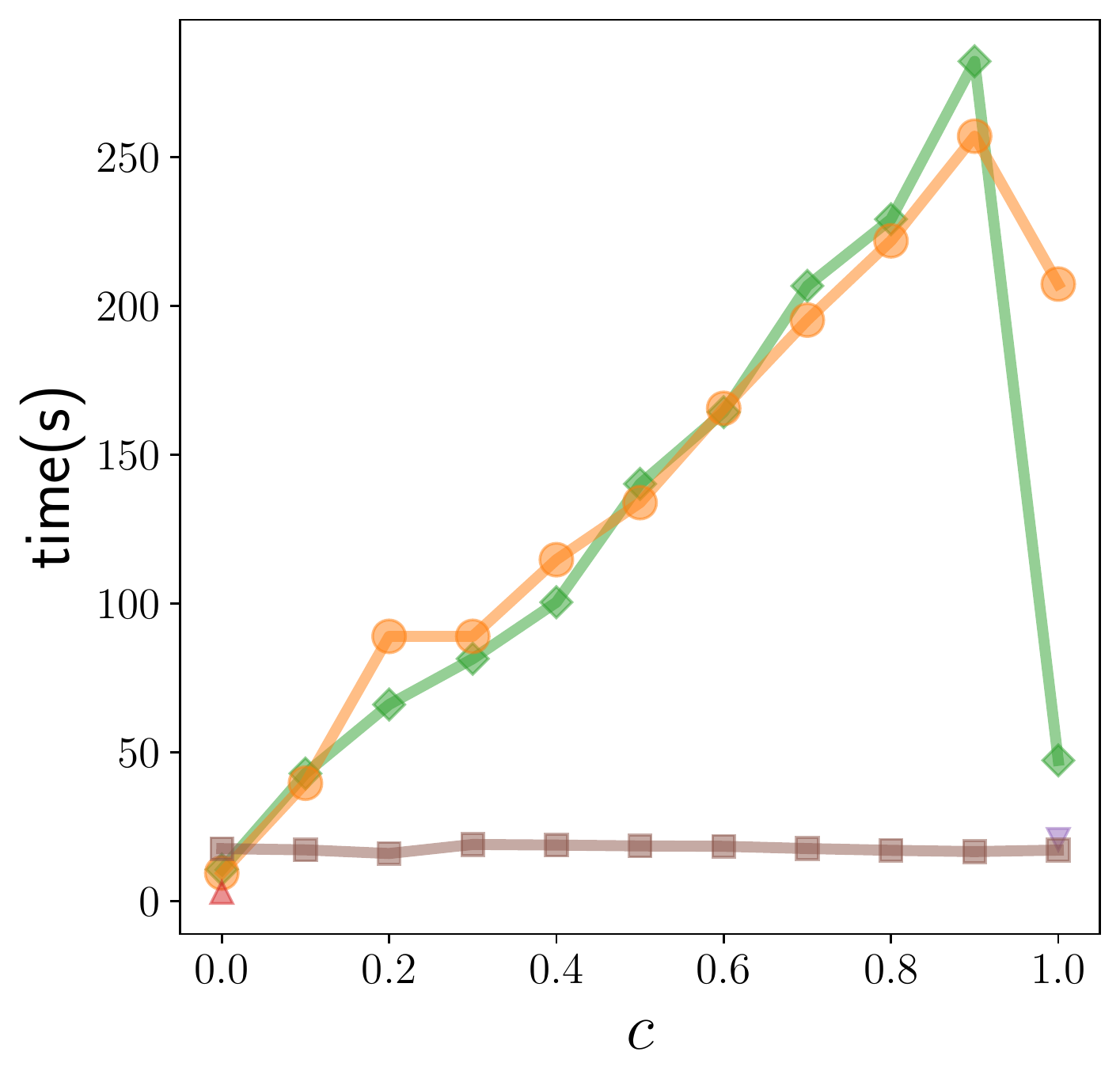}
}
\end{subfloat}
\begin{subfloat}[com-Youtube]{
\includegraphics[width=0.22\textwidth]{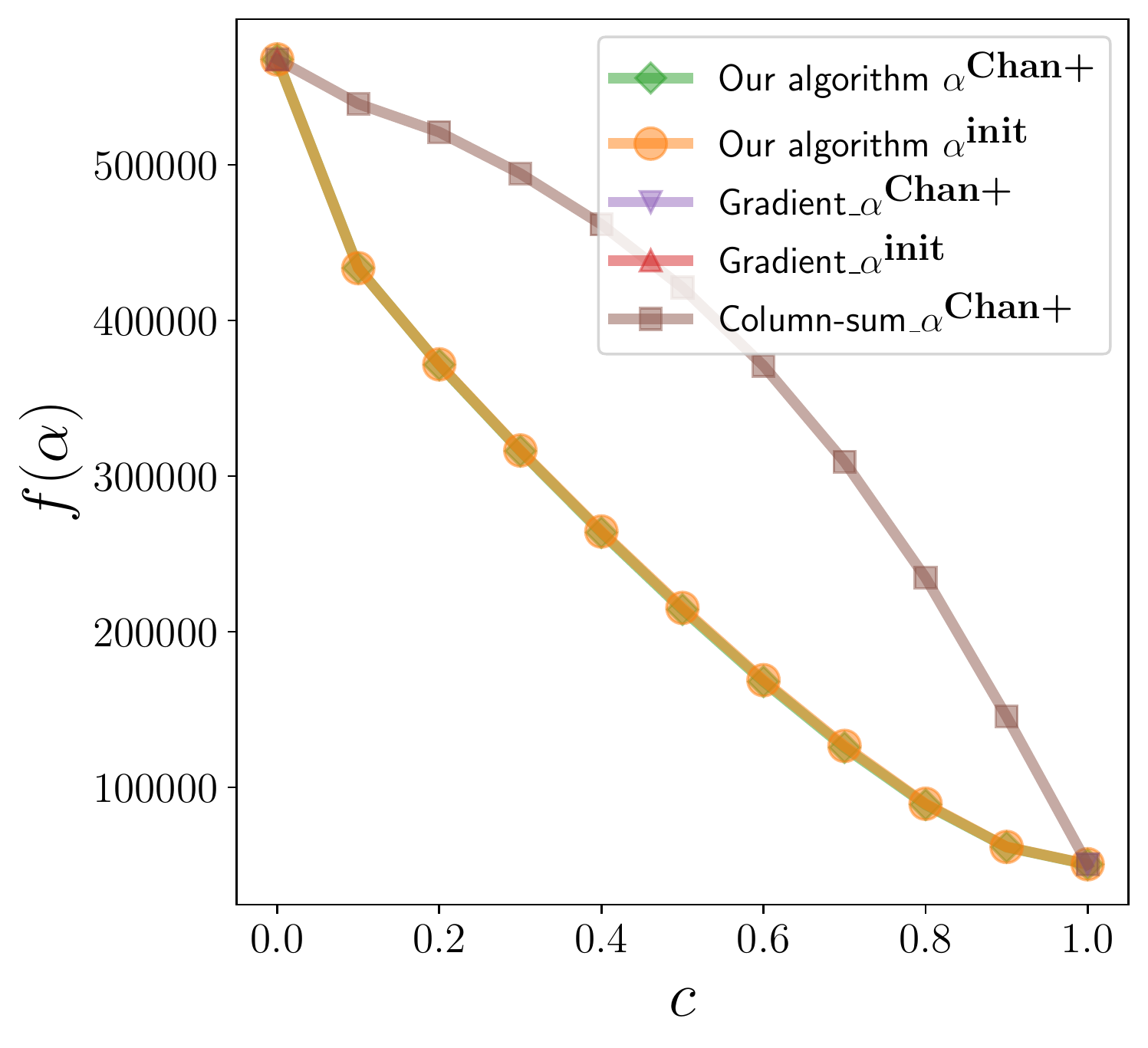}
\includegraphics[width=0.212\textwidth]{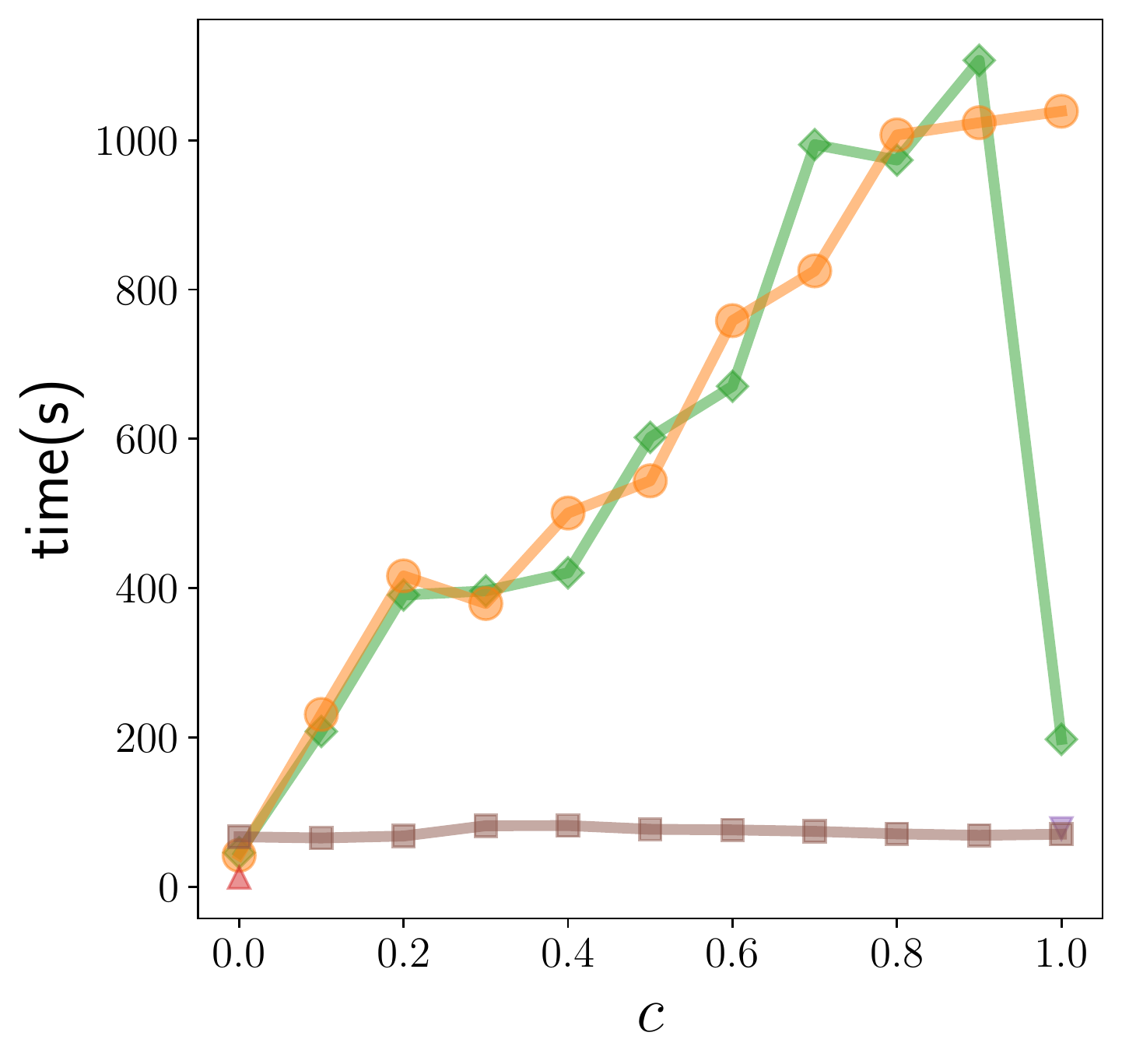}
}
\end{subfloat}
\caption{Results with $p=1$ for the last four graphs.}\label{fig:result_2}
\end{figure*}

\begin{figure*}
\begin{subfloat}[ca-GrQc]{
\includegraphics[width=0.22\textwidth]{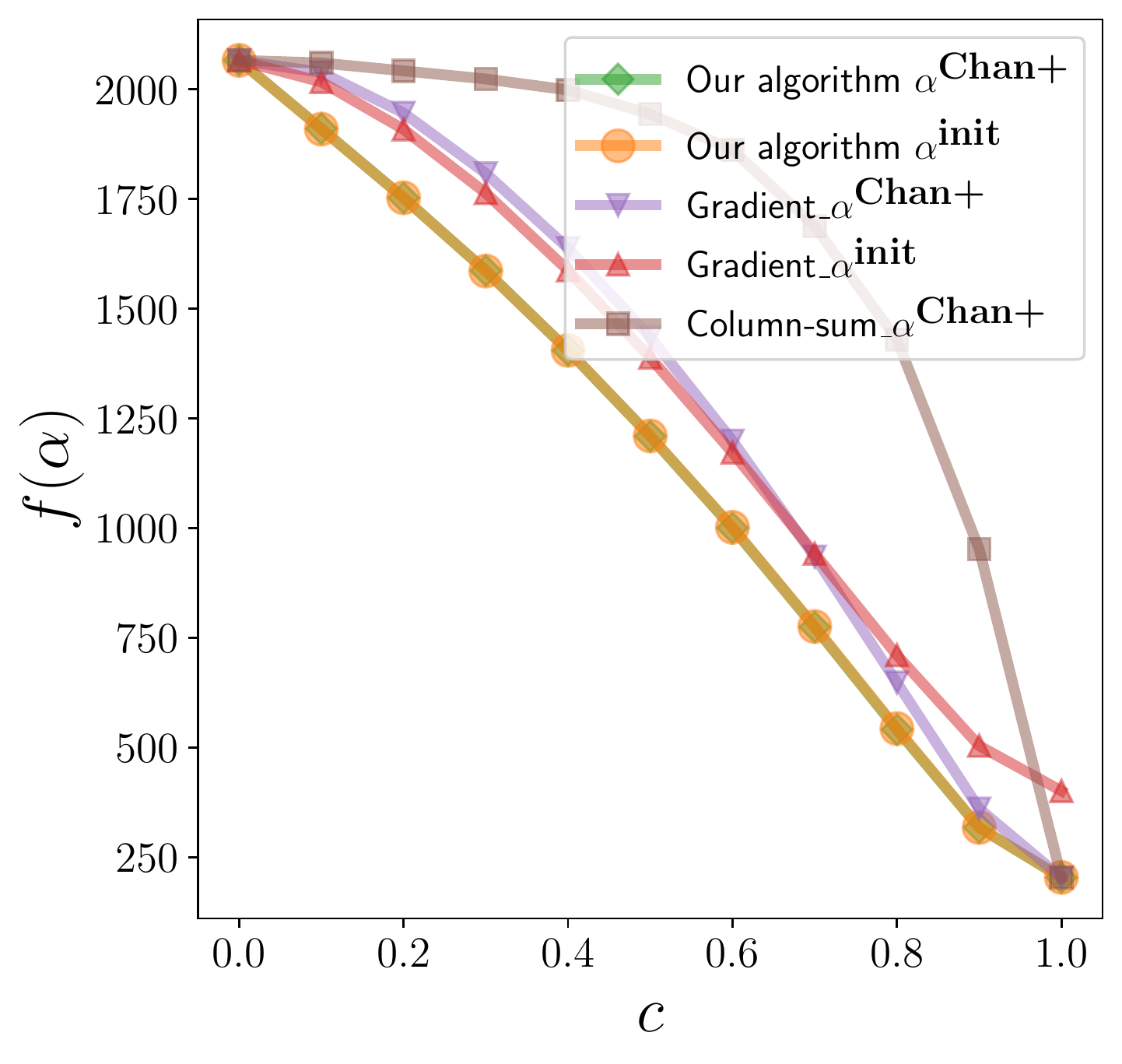}
\includegraphics[width=0.215\textwidth]{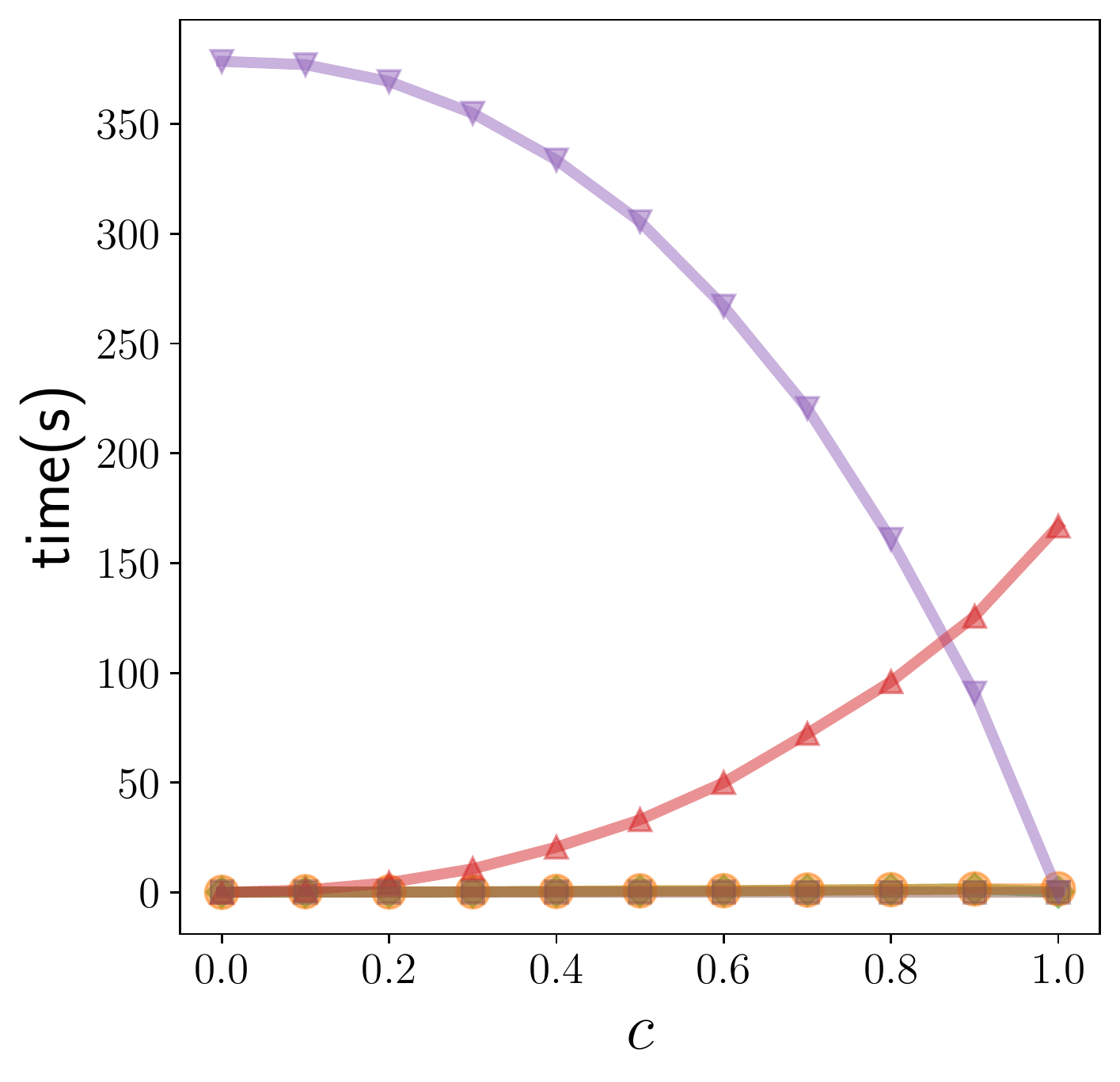}
}
\end{subfloat}
\begin{subfloat}[com-Youtube]{
\includegraphics[width=0.22\textwidth]{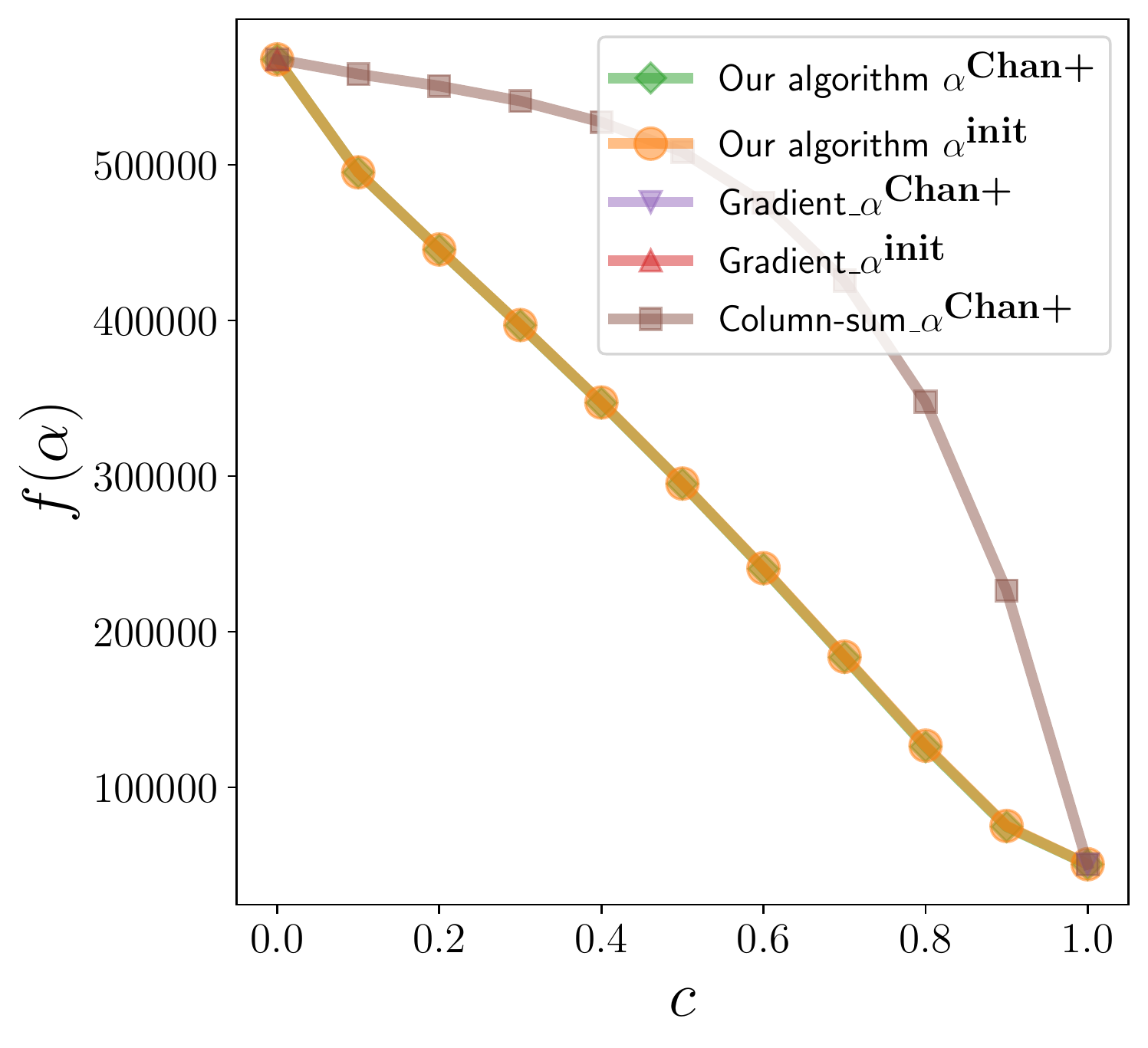}
\includegraphics[width=0.212\textwidth]{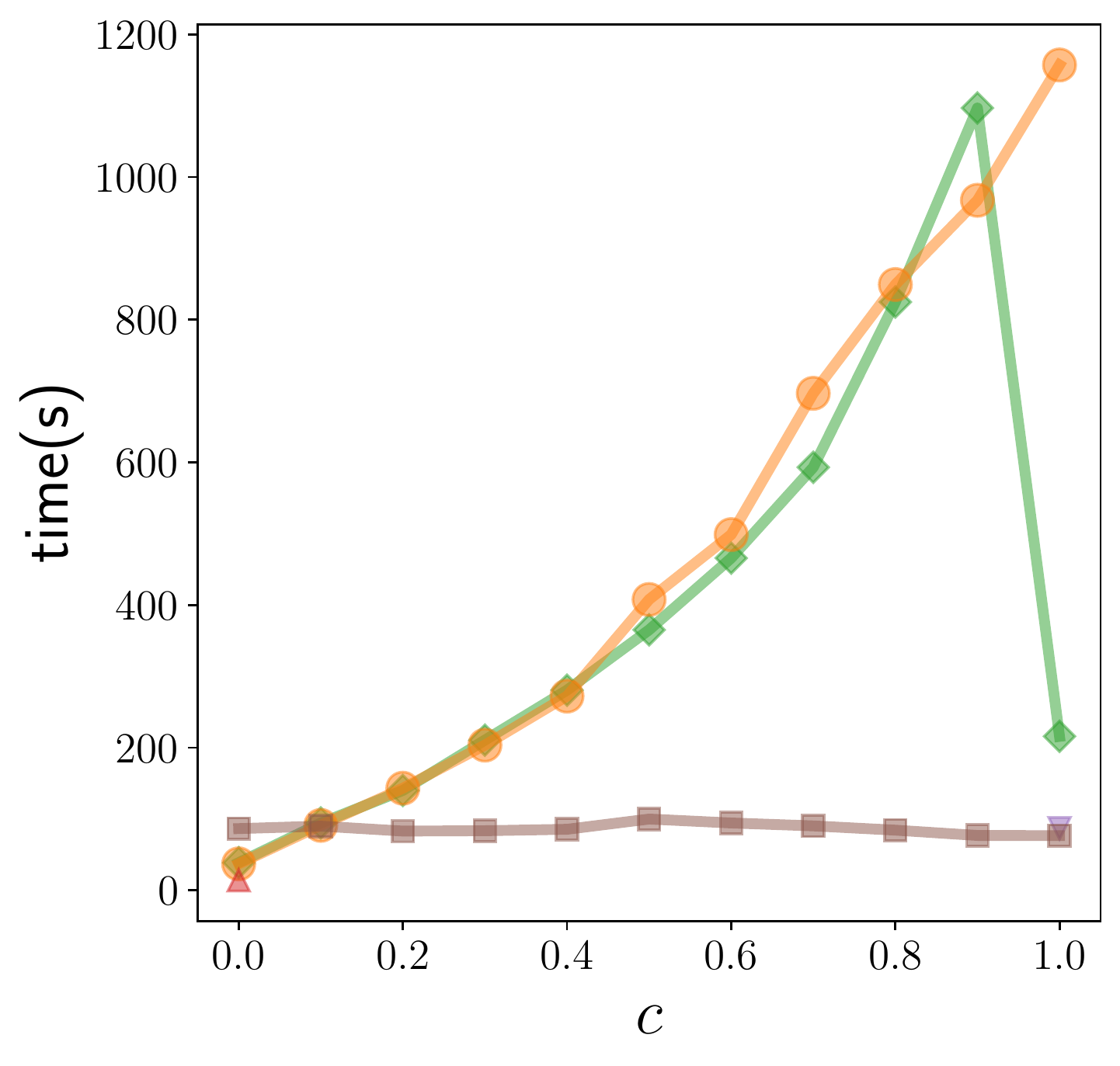}
}
\end{subfloat}
\caption{Results with $p=2$ for \textsf{ca-GrQc} and \textsf{com-Youtube}.}\label{fig:result_3}
\end{figure*}

\subsection{Results}

The results with $p=1$ are shown in Figures~\ref{fig:result_1} and \ref{fig:result_2}, 
where for each graph, the left figure shows the results of the quality of solutions while the right figure gives the results of the running time. 
As the instance generation procedure involves randomness, 
we generate ten instances for each graph, execute algorithms for each realization with each $c=0.0,0.1\dots, 1.0$, 
and plot the average value (in terms of both the objective function value and the running time). 

The time limit was set to 1,800 seconds. 
The running time of each algorithm includes that for computing an initial solution. 
If an algorithm exceeds the time limit for some realization for $G$ with some $c$, we did not plot the corresponding point at $c$. 
To save the time of experiments (in a reasonable way), we employed the following scheme for the executions of \textsf{Gradient\_$\bm{\alpha}^\mathrm{Chan+}$} and \textsf{Gradient\_$\bm{\alpha}^\mathrm{init}$}: 
As for \textsf{Gradient\_$\bm{\alpha}^\mathrm{Chan+}$}, we execute the algorithm for $c=1.0,0.9,\dots, 0.0$ in this order, 
and once we cannot plot the point for some $c$ (due to the time limit), we do not execute the algorithm for the following values of $c$. 
As for \textsf{Gradient\_$\bm{\alpha}^\mathrm{init}$}, we execute it in the opposite order of $c$ and apply the same rule. 

As can be seen in the figures, our proposed algorithm outperforms the baseline methods, irrespective of the choice of the initial solution. 
In terms of the quality of solutions, \textsf{Gradient\_$\bm{\alpha}^\mathrm{Chan+}$} has a comparable performance with our algorithm. 
However, due to its time complexity, the algorithm is not applicable to large instances. 
In fact, the algorithm exceeds the time limit for \textsf{ca-HepPh} and larger instances particularly for relatively small values of $c$. 
Although \textsf{Gradient\_$\bm{\alpha}^\mathrm{init}$} has a comparable performance with ours for small values of $c$, 
the performance becomes worse as the value of $c$ increases. 
Moreover, the algorithm has a similar issue as that of \textsf{Gradient\_$\bm{\alpha}^\mathrm{Chan+}$} in terms of the running time. 
\textsf{Column-sum\_$\bm{\alpha}^\mathrm{Chan+}$} is much faster than our proposed algorithms; 
it runs in less than 100 seconds even for the largest instance \textsf{com-Youtube}. 
However, the performance in terms of the quality of solution is much worse than ours and even the other two baselines. 
Our algorithm is reasonably applicable to the instances where there are millions of agents. 

Finally, the results with $p=2$ are depicted in Figure~\ref{fig:result_3}. 
Due to space limitations, we present the results only for \textsf{ca-GrQc}, the largest instance for which all algorithms ran in the time limit, 
and \textsf{com-Youtube}, the largest instance we used. 
As can be seen, the trend of the results is similar to that for the case of $p=1$; our proposed algorithm outperforms the baseline methods. 
Moreover, the difference between the performances (in terms of the quality of solutions) of our algorithm and the two baselines, 
\textsf{Gradient\_$\bm{\alpha}^\mathrm{Chan+}$} and \textsf{Gradient\_$\bm{\alpha}^\mathrm{init}$}, 
is more pronounced, compared with the case of $p=1$. 

\section{Conclusion}\label{sec:conclusion}
In this paper, we have studied an opinion optimization model that is able to limit the amount of changes of the susceptibility to persuasion in various forms. 
For our general model, we have designed a projected gradient method that is applicable to the case where there are millions of agents. 
Computational experiments using a variety of real-world social networks demonstrate the effectiveness of our proposed algorithm. 

There are several future directions. 
The most interesting one is to design a more scalable algorithm than ours. 
Is it possible to design a stochastic algorithm that requires the computation time just sublinear in the number of agents  
in each iteration?
Another direction is to extend our algorithm and analysis to the case of $p= 0$, where $\mathcal{C}$ is no longer convex. 
Finally, investigating the computational complexity of our model with general $p\geq 0$ is also interesting, 
as we only know the NP-hardness for $p=0$ and $1$~\cite{Abebe+18,Chan+21}. 

%

\begin{acks}
The authors thank Takayuki Okuno for helpful discussions on asymptotic convergence.
This work was supported by JSPS KAKENHI Grant Numbers 17H01699, 19H04069, and 19K20218. 
\end{acks}

\bibliographystyle{ACM-Reference-Format}
\bibliography{ref}

\end{document}